\documentclass{article}

\usepackage{arxiv}

\usepackage[utf8]{inputenc} % allow utf-8 input
\usepackage[T1]{fontenc}    % use 8-bit T1 fonts
\usepackage{hyperref}       % hyperlinks
\usepackage{url}            % simple URL typesetting
\usepackage{booktabs}       % professional-quality tables
\usepackage{amsfonts}       % blackboard math symbols
\usepackage{nicefrac}       % compact symbols for 1/2, etc.
\usepackage{microtype}      % microtypography
\usepackage{graphicx}

\newcommand{\numFamily}[1]{\mathbb{#1}}

\newcommand{\E}{\numFamily{E}}
\newcommand{\N}{\numFamily{N}}

\newcommand{\R}{\numFamily{R}}

\newcommand{\boldvec}[1]{ \boldsymbol{#1} }
\newcommand{\balpha}{{\ensuremath{\boldvec{\alpha}}}}
\newcommand{\bdelta}{{\ensuremath{\boldvec{\delta}}}}
\newcommand{\btheta}{{\ensuremath{\boldvec{\theta}}}}
\newcommand{\bphi}{{\ensuremath{\boldvec{\phi}}}}
\newcommand{\bpsi}{{\ensuremath{\boldvec{\psi}}}}
\newcommand{\btau}{{\ensuremath{\boldvec{\tau}}}}

\newcommand{\bp}{{\ensuremath{\boldvec{p}}}}
\newcommand{\bq}{{\ensuremath{\boldvec{q}}}}

\newcommand{\bT}{{\ensuremath{\boldvec{T}}}}
\newcommand{\bu}{{\ensuremath{\boldvec{u}}}}
\newcommand{\bv}{{\ensuremath{\boldvec{v}}}}
\newcommand{\bx}{{\ensuremath{\boldvec{x}}}}
\newcommand{\by}{{\ensuremath{\boldvec{y}}}}

\newcommand{\Dir}{\mathop{\textrm{Dir}}}
\newcommand{\Mult}{\mathop{\textrm{Mult}}}
\newcommand{\logit}{\mathop{\textrm{logit}}}

\usepackage{amsthm} % permit unnumbered theorems
\newtheorem{theorem}{Theorem} 
\newtheorem{theorem*}{Theorem} 
\newtheorem{prop}{Proposition}
\newtheorem*{prop*}{Proposition}
\newtheorem*{corollary*}{Corollary} 
\newtheorem*{remark*}{Remark}
\newtheorem{lemma}{Lemma}
\newtheorem*{lemma*}{Lemma} 
\newtheorem*{definition*}{Definition} 
\newtheorem*{algorithm*}{Algorithm} 

\usepackage[dvipsnames]{xcolor} % colors

\usepackage[
	backend=biber,
	citestyle=authoryear,
	sorting=nyt,
	natbib=true,
	maxnames=2,
	uniquename=false
]{biblatex}
\usepackage{mathtools}
\usepackage[capitalise]{cleveref} %clever ref, has to be included last

\title{Modelling birdsong transmission with methods from molecular sequence analysis}

\author{
	Anthony Kwong \\
	Department of Mathematics\\
	University of Manchester\\
	Manchester, UK \\
	\texttt{shingyan.kwong@manchester.ac.uk} \\
	\And
	\href{https://orcid.org/0000-0002-5004-7195}{\includegraphics[scale=0.06]{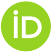}\hspace{1mm}Mark Muldoon}\thanks{Author to whom any correspondence should be addressed.}  \\
	Department of Mathematics\\
	University of Manchester\\
	Manchester, UK \\
	\texttt{mark.muldoon@manchester.ac.uk} \\
}

\date{\today}

\renewcommand{\headeright}{}
\renewcommand{\undertitle}{}
\renewcommand{\shorttitle}{Modelling birdsong transmission}

\hypersetup{
pdftitle={Modelling birdsong transmission with methods from molecular sequence analysis},
pdfsubject={q-bio.QM, stat.CO},
pdfauthor={Anthony~Kwong, Mark~Muldoon},
pdfkeywords={birdsong \and social learning \and maximum marginal likelihood estimation \and 
bridge sampling \and stochastic gradient ascent
},
}

\begin{document}
\maketitle

\begin{abstract}
In many species of songbirds, juvenile males learn their songs from adult
male tutors. In this paper we formulate a novel Markov model for birdsong
transmission developed by analogy with models used in biological sequence
analysis. We fit the  model using the recently developed Interacting
Particle Langevin Algorithm (IPLA) of \cite{Akyildiz:2023nh} and analyse a
collection of songs from Java sparrows (\emph{Lonchura oryzivora})
originally recorded and studied by Masayo~Soma and her collaborators
\citep{soma2011social,Lewis:2021hj,Lewis:2023gp}. The model proves to have
limited predictive power for a number of natural problems associated
with song transmission in Java sparrows and we propose reasons for
this, including the well-established faithfulness of song-learning and the
comparative brevity of Java sparrow songs.	
\end{abstract}

% keywords can be removed
\keywords{birdsong \and social learning \and maximum marginal likelihood estimation \and 
bridge sampling \and stochastic gradient ascent}

\section{Introduction}

In many species of songbirds, juvenile males learn their songs from adult male tutors.
In this paper we ask whether tutor-to-pupil song transmission can be studied by analogy with molecular sequence evolution. Where songs consist of a sequence of notes, genomes are sequences of nucleotides. Like nucleotides, notes can change from one type to another due to learning errors (i.e. point mutations). Individual notes or even song segments can be inserted, deleted, or duplicated, just like nucleotides and motifs in genomes. Hence, tools from molecular phylogenetics may provide a foundation for modelling birdsong transmission. However, there are important differences between genomes and birdsong. For example, individuals have a single genome, but a bird may sing many different songs. Although there are only four nucleotide bases, birds can sing many different note types. Moreover, molecular phylogenetics is an advanced field where estimates of mutation rates and substitution matrices are available, while birdsong transmission is an emerging subject where such tools are not yet available. These differences necessitate the development of novel statistical approaches.

\subsection{Java sparrow songs}

Java sparrows (\emph{Lonchura oryzivora}) are a highly vocal estrildid finch native to the island of Java, Indonesia (\cite{Kagawa:2013hn}). Male birds typically sing between 2 and 8 note types and their songs often include repetitive syllables, trills, and non-trill notes. Previous studies have established that male juveniles (\emph{pupils}) learn their songs from their so-called social fathers, whom we will call \emph{tutors} in this paper. (\cite{soma2011social}). Learning fidelity is high: pupils learn their tutors' note repertoires, note sequences (song structure) and temporal features (\cite{Lewis:2021hj, Lewis:2023gp}). The strong effect of social learning provides the rationale for creating a data-driven, model for the transmission of note sequences.

\subsection{Overview of the data}

We study a dataset published in (\cite{Lewis:2021hj}) consisting of whole song recordings from a captive population of Java sparrows, reared by the Soma lab at the University of Hokkaido. The birds have been selectively tutored, so that almost all pupil-tutor relationships are known. Recordings were made in a soundproof chamber with each bird singing in isolation and all recordings for a given bird were made within a single week
(time point). For each bird, when more than 10 full song recordings were available from a single time point, 10 were randomly selected for inclusion. If fewer than 10 recordings were available at all time points, all songs from one time point were used, prioritising recordings made when the bird was approximately 2–5 years old. Full details of the data collection procedures are provided in Lewis, Soma et al. (2021). In total, there are 676 songs recorded from 73 individuals (mean 9.3 songs per individual, range 3 to 10), including 58 pupil-tutor pairs. The birds came from 10 different social lineages, named with arbitrary colours by the Soma lab. A social lineage is a chain of birds all of whose members can trace their songs back to the same founder: see Figure~\ref{fig:LineageForest}. Lineages varied from 1 to 4 generations in length. 

\begin{figure}
	\begin{center} 
		\includegraphics[width=0.95\columnwidth]{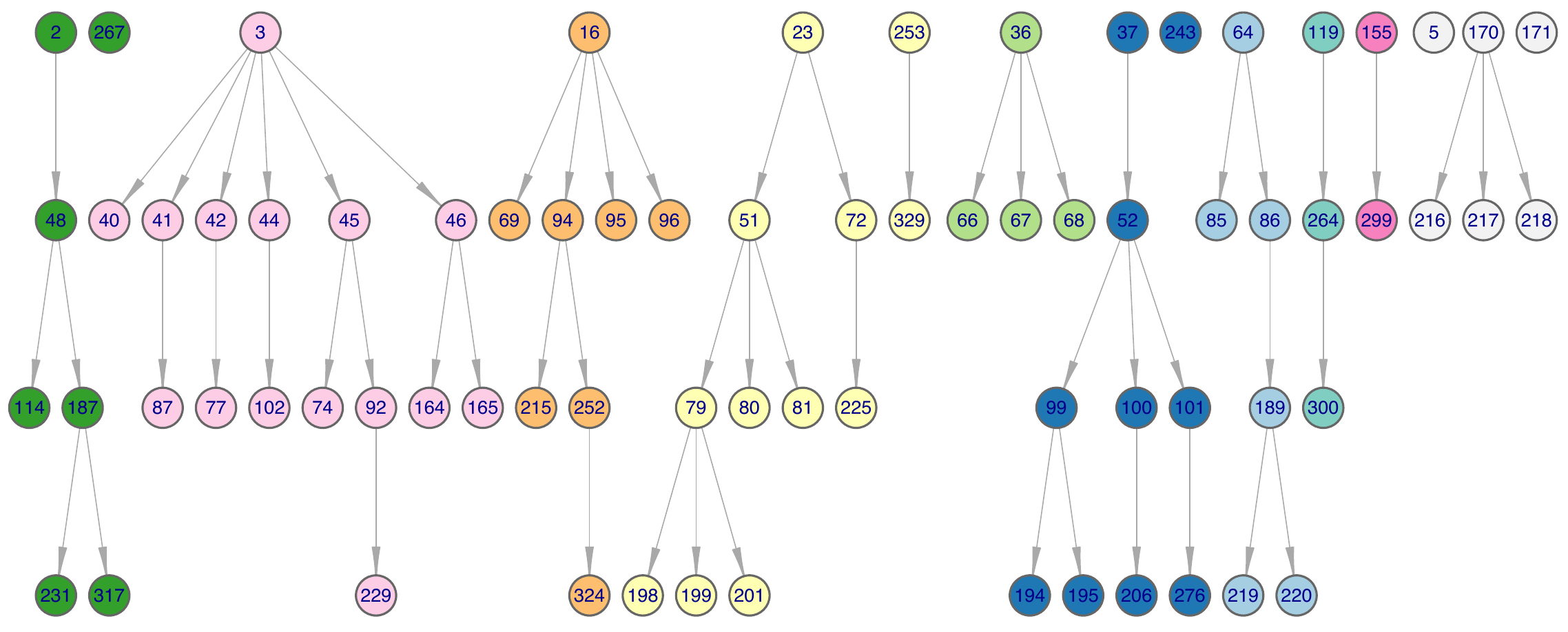}
	\end{center}
	\caption{ 
		Song-transmission lineages represented in the data from (\cite{Lewis:2021hj}).
		Nodes are labelled with numerical bird IDs, while directed edges point from
		tutor to pupil. Lineages are socially connected groups of birds linked through song transmission: nodes for birds in the same lineage are shown in same colour. Birds 267, 243, 5 and 171 do not have any pupils and their tutors were unknown, so Soma assigned them to lineages based on source location. For example, birds 5 and 171 were purchased from the same pet store as bird 170. }
	\label{fig:LineageForest}
\end{figure}

\begin{figure}
	\begin{center} 
		% See https://tex.stackexchange.com/questions/16942/difference-between-textwidth-linewidth-and-hsize
		\includegraphics[width=0.9\columnwidth]{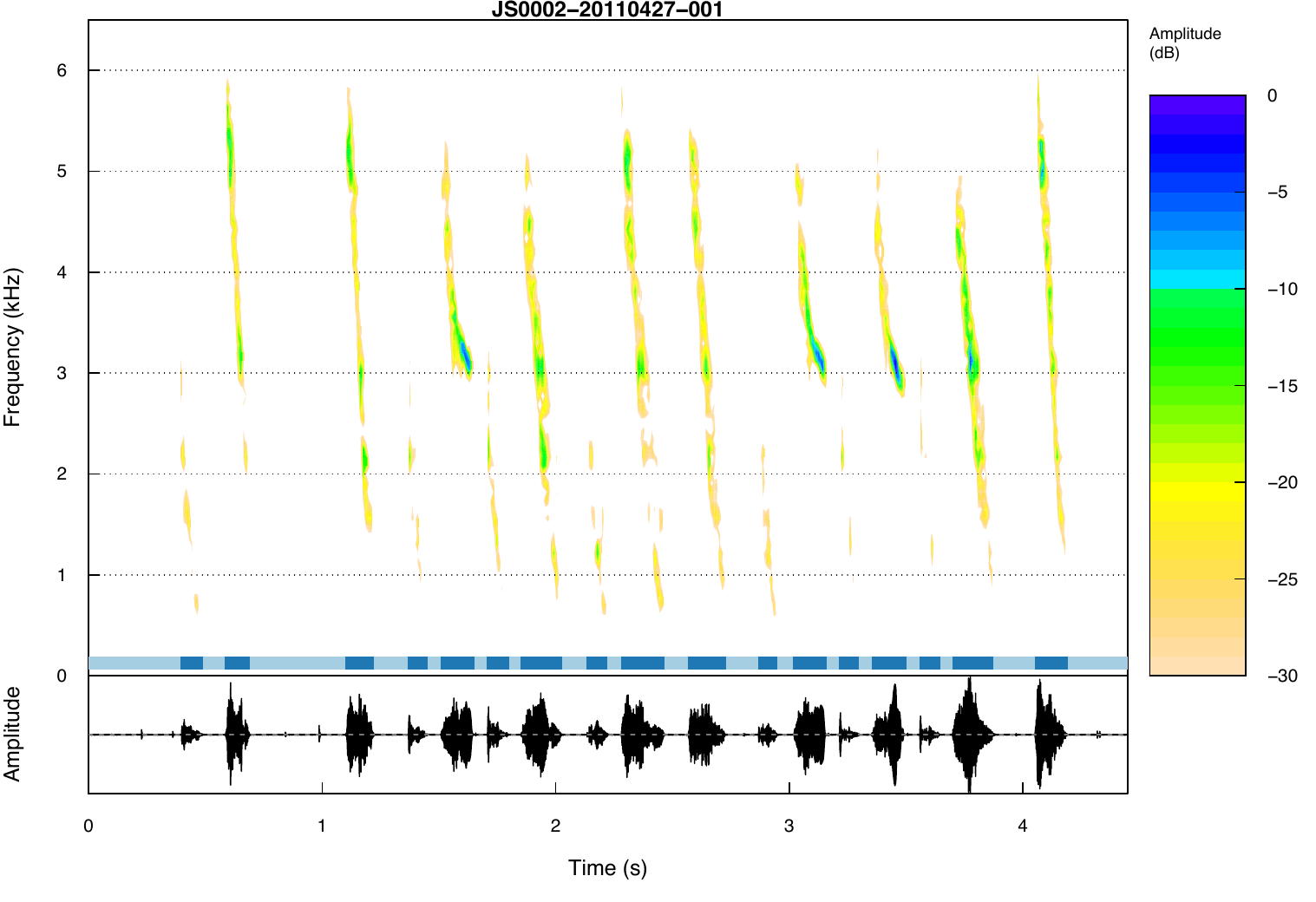}
	\end{center}
	\caption{ 
		Example spectrogram of one song, sung by JS0002 from the Java sparrow dataset (\cite{Lewis:2021hj}). Every vertical strip in the diagram holds a heatmap representing the power spectrum for a 512-sample ($\approx 0.0116$ secs) temporal interval. The darker regions in the  blue bar toward the bottom of the plot indicate the intervals from which the individual notes were extracted. 
	}
	\label{fig:specgram_sample}
\end{figure}

\Cref{fig:specgram_sample} is a visual representation of one song in the Java sparrow dataset. The coloured waves are the individual sound elements (i.e. notes) which are separated by intervals of near-silence. \textcite{Lewis:2021hj} inspected such spectrograms and excised notes using the software tool Koe (\cite{Fukuzawa:2020cd}). They then manually classified the notes based on their shape and spectral features: a full discussion of the note classification procedure is provided in \textcite{Lewis:2021hj}. In total, Lewis et al. classified 22,972 notes into 16 note classes: see \Cref{fig:noteclass_sample} for typical examples of the note types. There is now software to computationally excise and classify notes with high reliability using tools such as TweetyNet (\cite{Cohen:2022xn}). By excising and classifying the notes, Lewis et al. converted each song into a note sequence. We generated multiple sequence alignments of note sequences from tutors and pupils using profile hidden Markov models 
(\cite{Kwong:2025dv}: in preparation). 

\begin{figure}
	\begin{center} 		\includegraphics[width=0.9\columnwidth]{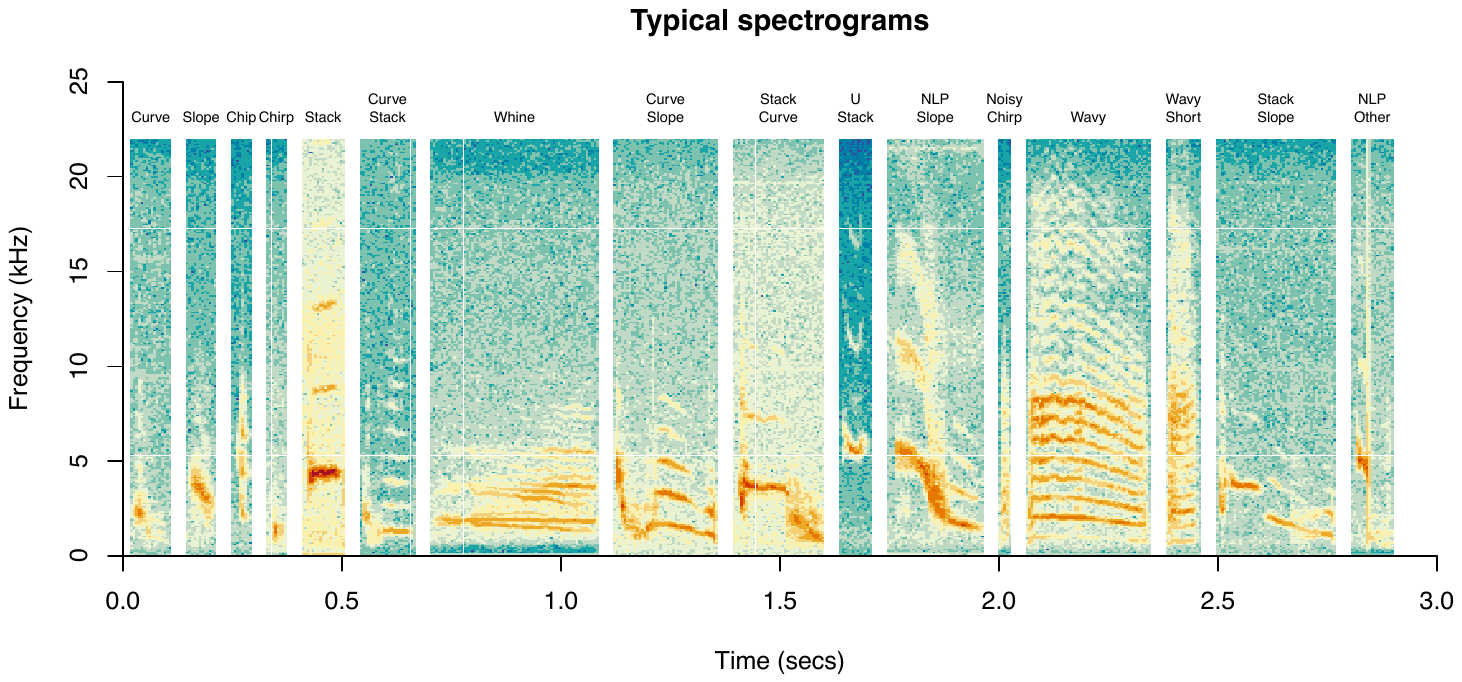}
	\end{center}
	\caption{ 
    Example spectrograms of the 16 note classes identified by \cite{Lewis:2021hj} in the Java sparrow dataset. Note that the scale used to draw the heatmap varies from note to note. 
	}
	\label{fig:noteclass_sample}
\end{figure}

\subsection{Reduction of songs to count data}
\Cref{fig:alignment} illustrates the kind of data that we
analyze. It consists of a set of aligned note sequences from a tutor and a
pupil. We restricted attention to those columns of the alignment where both
birds have notes in at least half of their songs (as opposed to gap characters). Next, we reduced the alignment to a pair of matrices of counts:
\begin{center}
	\begin{tabular}{c|l}
		$x_{ij}$ & Number of times tutor sang note $i$ in position $j$ of the alignment \\
		$y_{ij}$ & Number of times pupil sang note $i$ in position $j$ of the alignment\end{tabular}
\end{center}
For example, we write $\bx_j = (x_{1j}, \dots, x_{dj})$ to indicate the vector 
of the tutor's counts at the $j$-th aligned position and $\by_j$ for the pupil's 
counts. We will do the same thing for all pupil-tutor pairs and so the index $j$ 
will run across \emph{all} aligned positions, both within the
alignment for a single known pupil/tutor pair and across all such pairs.
In Section~\ref{sec:NotionsAndNotation} below, we develop a probabilistic
model for the generation of such counts.

\begin{figure}
	\begin{center} 
		% See https://tex.stackexchange.com/questions/16942/difference-between-textwidth-linewidth-and-hsize
		\includegraphics[width=\columnwidth]{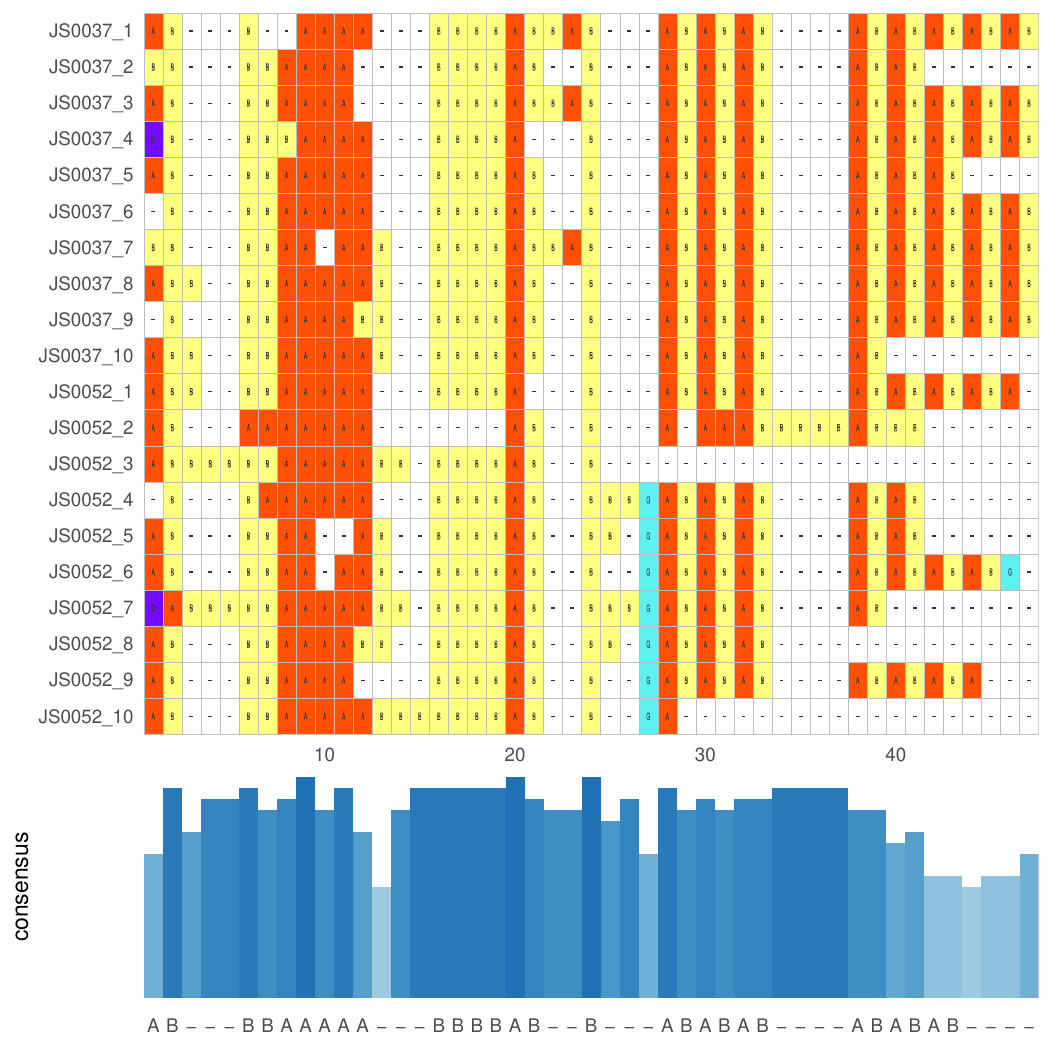}
	\end{center}
	\caption{ 
		A typical alignment produced using profile hidden Markov models
		(see Chapter 5 \textcite{Durbin:1998yu} for details of this approach).
		Each row corresponds to a recorded song while the 
		differently coloured blocks represent different note types and 
		are labelled with letters: A is the most commonly sung note 
		(across all recordings), B the second most commonly sung note, and so on.
		The 20 songs aligned here come from a pupil-tutor pair, birds JS0037 (the tutor) 
		and JS0052 (the pupil).
	}
	\label{fig:alignment}
\end{figure}

\subsection{Outline of the paper}
The remainder of the paper is organised as follows. Section~\ref{sec:NotionsAndNotation} lays out the mathematical foundations of our model and the associated inference problems.
In particular, in Section~\ref{sec:generativeModel} we develop a probabilistic model for birdsong transmission using an analogy to models used in molecular phylogenetics. A key ingredient is the \emph{tranmission matrix} $T$, whose element
$T_{rs}$ is the probability that a pupil will sing note $r$ given that his tutor
sang note $s$. In Section~\ref{sec:BayesianInferenceForT} we formulate a Bayesian inference problem for the model's parameters and in Sections~\ref{sec:introMMLE}---\ref{sec:stickBreaking}
we explain how to use the recently developed Interacting Particle Langevin Algorithm (IPLA) of \cite{Akyildiz:2023nh}---which is essentially an optimisation algorithm---to estimate $T$. 

One of the attractions of the IPLA is that, subject to certain restrictions on the optimisation target, it comes with strong convergence guarantees. In Section~\ref{sec:smoothness} we ask whether our optimisation target, a certain marginal likelihood, satisfies these conditions and conclude that it does not: full proofs are supplied in an Appendix. To address this problem, we did numerical 
experiments in which we applied the IPLA to synthetic data generated from our model. We report on these in Section~\ref{sec:syntheticData}, where we find that despite the lack of theoretical guarantees, the inference problem shows no evidence of multimodality and we are able to recover the transmission matrix. Bolstered by these observations, we fit a transmission matrix for the birdsong data: these results are reported in Sections~\ref{sec:realData} and \ref{sec:trans_mat}. All these calculations rely on codes that we developed ourselves and ran in R \citep{R-Core-Team:2025zi}, accelerated with the \texttt{Rcpp} package \citep{Eddelbuettel:2025yo,Eddelbuettel:2011hy}.

Armed with the fitted transmission matrix, we then used the Hamiltonian Monte Carlo package Stan \citep{Carpenter:2017rm} and the \texttt{bridgesampling} R package \citep{Gronau:2020fb} to estimate the posterior probabilities that each possible pair of birds had a pupil-tutor relationship: these results are summarised in Section~\ref{sec:evidence}. The paper concludes with a discussion in Section~\ref{sec:conclusion} and an Appendix providing detailed proofs. 

\section{Notions and notation}
\label{sec:NotionsAndNotation}

Here, we define terms and introduce tools that we will use to specify our probabilistic model.\newline

The \emph{$d$-dimensional simplex} $\Delta_d$  is the set
\begin{displaymath}
	\Delta_d = \left\{ (p_1, \dots, p_d) \in \R^d \; \biggl|  \; p_i \geq 0 \;\; \mbox{and} \;\; \sum_{i=1}^d p_i = 1 \right\}.
\end{displaymath}
One can think of points in the interior of $\Delta_d$ as discrete probability 
distributions over the numbers $1, \dots, d$.

The \emph{multinomial distribution} is a probability distribution over 
vectors of counts $\by \in \N^d$ that is parameterised by the total number 
of counts $N = \sum_{i=1}^d y_i$ and a vector of probabilities $\bp \in \Delta_d$.
Its probability mass function is
\begin{equation}
	P( \by \, | \, \bp, N ) \equiv \Mult( \by \, | \, \bp ) =
	\frac{N!}{\prod_{i=1}^d y_i !} \left( \prod_{i=1}^d p_i^{y_i} \right). 
\label{eqn:multinomialDef}
\end{equation}

The \emph{Dirichlet distribution} is a probability distribution over 
$\Delta_d$ and is parameterised by a vector of positive \emph{shape
parameters} $\balpha \in \R_+^d$. Its probability density function is
\begin{equation}
	P( \bp \, | \, \balpha  ) \equiv \Dir( \bp\ \, | \, \balpha  ) =
	 \frac{\Gamma( \alpha ) }{\prod_{i=1}^d \Gamma(\alpha_i)}  \left( \prod_{i=1}^d p_{i}^{\alpha_i - 1} \right),
\label{eqn:DirichletDef}
\end{equation}
where $\alpha = \sum_{i=1}^d \alpha_i$ is the sum of the shape
parameters.

\subsection{A probabilistic generative model}
\label{sec:generativeModel}

If we then introduce the following unobserved (and, indeed, unobservable) quantities
\begin{center}
	\begin{tabular}{c|l}
		$p_{ij}$ & Latent probability that the tutor will sing note $i$ in position $j$ \\
		$q_{ij}$ & Latent probability that the pupil will sing note $i$ in position $j$ \\
	\end{tabular}
\end{center}
then the likelihood for the count data is:
\begin{equation}
	L = \prod_j \Mult(\bx_j \, | \, \bp_j ) \times \Mult(\by_j \, | \, \bq_j )
\label{eqn:firtLikelihoodExpr}
\end{equation}
where the product ranges over all aligned positions for all pupil-tutor pairs. 
Here $\bp_j = (p_{1j}, \dots, p_{dj})$ is the tutor's vector of note-usage 
probabilities at the $j$-th aligned position and $\bq_j$ is the pupil's. 

In the spirit of substitution models used in phylogenetics,
we assume that the pupil's probabilities $\bq_j$ are 
related to the tutor's $\bp_j$ by a single, fixed \emph{transmission matrix}
$T$ that is shared across all aligned positions, so that
\begin{equation}
	\bq_j = T \bp_j.
\label{eqn:qTp}
\end{equation}
Here $T$ is a $d \times d$ matrix and
\begin{equation}
	T_{rs} = P(\mbox{pupil sings note $r$} \, | \, \mbox{tutor sings note $s$} ).
\label{eqn:TmatDef}
\end{equation} 
As the pupil has to sing \emph{something}, the columns of $T$ must satisfy:
\begin{equation}
	\sum_{r=1}^d T_{rs} \; = \; 
	\sum_{r=1}^d P( \mbox{pupil sings $r$} \; |\; \mbox{tutor sings $s$} )
	\; = \; 1.
\label{eqn:colTsumsToOne}
\end{equation}
Hence, we obtain a new expression for the log-likelihood:
\begin{equation}
	\mathcal{L} = \log(L) = \sum_j \log\left(\Mult(\bx_j \, | \, \bp_j )\right) 
		+ \log\left(\Mult(\by_j \, | \, T \bp_j )\right),
\label{eqn:loglike}
\end{equation}
where, as above, the sum ranges over all aligned positions for all pupil-tutor pairs.

\subsection{Bayesian inference}
\label{sec:BayesianInferenceForT}

Given we have the aligned note sequences and a transmission model, we need only 
specify a set of priors to set up a Bayesian inference problem.
Let $D$ denote the data (the counts $\bx_j$ and $\by_j$ from the alignment) 
and $\btheta$ be the parameters ($T$ and the latent probabilities $\bp_j$). 
Then we write

\begin{description}
	\item[$P( D \, | \, \btheta )$] the \emph{likelihood}.
	$P( D \, | \, \btheta ) = \prod_j \Mult(\bx_j \, | \, \bp_j ) \times \Mult(\by_j \, | \, T \bp_j )$.
	
	\item[$P( \btheta )$] the \emph{prior on the parameters}. We use 
	$d$-dimensional Dirichlet distributions here: one for every column of $T$ and one
	for each $\bp_j$.
	
	\item[$P( \btheta \, | \, D )$] the \emph{posterior distribution
	over $\btheta$}, $P( \btheta \, | \, D ) \propto P( D \, | \, \btheta )P( \btheta ).$

	\item[$P(D)$] the \emph{marginal likelihood}. This is a value
	from the marginal distribution obtained by integrating $\btheta$
	out of $P(D, \btheta)$
	\begin{displaymath}
		P(D) = \int P(D \,|\, \btheta) P(\btheta) \, d\btheta.
	\end{displaymath}
	
\end{description}
 
\subsection{Maximal marginal likelihood estimation of $T$}
\label{sec:introMMLE}

It would be straightforward to use Hamiltonian Monte Carlo (HMC) to draw
samples from the posterior over $T$ and the $\bp_j$. But as we are only interested in $T$, the $\bp_j$ are nuisance parameters. Ideally, we would compute or sample from the marginal distribution obtained by integrating them out of the posterior:
\begin{equation}
	P(T \; | \; D) = \int P(T, \bp \; | \; D) \, d\bp
	= \frac{1}{P(D)} \int P(D \; | \; T, \bp ) P(T, \bp) \, d\bp.
\label{eqn:marginalForT}
\end{equation}
where $\bp$ is a vector of \emph{all} the latent probabilities of
note-usage $\bp_j$ and $T$ is the transmission matrix. 

If we could compute $P(T \; | \; D)$, we could then maximise it with respect to the elements of $T$ to obtain a \emph{maximal marginal likelihood estimate (MMLE)} for $T$. And since the marginal likelihood $P(D)$ does not depend on $T$, it would be sufficient to maximise
\begin{equation}
	 \int P(D \; | \; T, \bp ) P(T, \bp) \, d\bp.
\label{eqn:IPLA_Target}
\end{equation}
Unfortunately the integral in \eqref{eqn:IPLA_Target} is intractable and so we are
forced to resort to numerics, using a recently-developed algorithm described below.

\subsubsection{MMLE with stochastic differential equations}
Recent work of \textcite{Akyildiz:2023nh}, inspired by the
methods of \textcite{Kuntz:2023rp}, suggests a diffusion-based approach to maximal marginal likelihood estimation, the \emph{interacting particle Langevin algorithm} (IPLA), 
that is a member of a family 
of algorithms that generalise the EM-algorithm using ideas that originated in \textcite{Neal:1998jw}. The setup includes observed data denoted by $y$ (for us this would be the count data), parameters that one wishes to estimate denoted by $\theta$ and latent variables, denoted by $x$, that one wishes to marginalise-out. They begin by defining the negative log-likelihood,
\begin{equation}
    U(\theta,x) \equiv -\log( p_\theta(x,y)),
\end{equation}
where $p_\btheta(x,.)$ is the joint probability distribution of the latent variables $x$ and observed data $y$ given the parameters $\theta$. The goal of IPLA is to maximise the marginal likelihood:

\begin{equation}
	k(\theta) \equiv p_\theta(y) = \int p_\theta(x,y) dx = \int e^{-U(\theta, x )} \, dx
\label{eqn:kDefIPLA}
\end{equation}

\cite{Akyildiz:2023nh} propose to maximise $k(\theta)$ via gradient ascent, approximating $\nabla k(\theta)$ by averaging over $N$ of independent realisations of a stochastic differential equation (SDE) on the full parameter space:
\begin{align}
	d\theta_t^N & = -\frac{1}{N} \sum_{j=1}^N \nabla_\theta U(\theta_t, X_t^{j,N}) \, dt + \sqrt{\frac{2}{N}} \, dB_t^{0,N} \nonumber \\
	dX_t^{i,N} & = -\nabla_x U(\theta_t^N, X_t^{i,N}) \, dt + \sqrt{2} \, dB_t^{i,N}
\label{eqn:AkyilidizSDEs}
\end{align}
for $i = 1, \dots , N.$ Here
$X^{i,N}_t$ are the $N$ particles used for estimating the gradient and
$\{ (B_t^{i,N})_{t \geq 0} \}^N_{i=0}$ is a family independent Brownian motions.
% $\sqrt{\frac{2}{N}} \, dB_t^{0,N}$ is a small, noise term. 
In practice, Akyildiz et al. use an Euler-Maruyama discretisation to generate approximate realisations. Given 
$(\theta_0, X_0^{1:N,N}) \in \mathbb{R}^{d_\theta} \times (\mathbb{R}^{d_x})^{\otimes N}$ and for any $n \in \mathbb{N}$
\begin{align}
	\theta_{n+1} & = \theta_n -\frac{\gamma}{N} \sum_{j=1}^N \nabla_T U(\theta_n, X_n^{j,N}) + \sqrt{\frac{2 \gamma}{N}} \, \xi_{n+1}^{0,N} \nonumber \\
	X^{i,N}_{n+1} & = X^{i,N}_{n} -\gamma \nabla_\bp U(\theta_n, X_n^{i,N}) + \sqrt{2 \gamma} \, \xi_{n+1}^{i,N} 
\label{eqn:AkyilidizEulerMaruyama}
\end{align}
where $\gamma > 0$ is a step size and $\xi_{n}^{i,N} = B_{n\gamma}^{i,N} - B_{(n-1)\gamma}^{i,N}$ for $0 \leq i \leq N$
are i.i.d. samples from the standard normal.

An advantage of the IPLA is that, subject to certain smoothness and convexity assumptions on $U$, one can obtain provable bounds on:
\begin{itemize}
	\item the range of step sizes $\gamma$ that will yield convergence to
	 $T^\star$, the MMLE for $T$;
	
	\item the expected deviation of $T_n$ from the MMLE, $T^\star$: 
	\begin{displaymath}
		\E\left[\| T_n - T^\star \|^2\right]^{1/2}
	\end{displaymath}
\end{itemize} 
The second set of bounds depends on the step-number $n$ and so also provides a 
bound on the rate of convergence of the algorithm.
These smoothness and convexity assumptions, which are discussed---and, where possible, verified---for our application in Section~\ref{sec:smoothness} below, also ensure that 
\begin{displaymath}
	\nabla_T k(T) \; = \;  \nabla_T \int  e^{-U(T, \bp )} \, d\bp
		\; = \; - \int \nabla_T  U(T, \bp ) e^{-U(T, \bp )} \, d\bp.
\end{displaymath}

\subsubsection{The IPLA for birdsong transmission}
For our application the role of Akyldiz's data $y$ is played by the count data, the latent variables $x$ are our $\bp_j$, the collection of latent note probabilities, and $\theta$ is the transmission matrix $T$. The optimisation target, a negative log-posterior, is
\begin{align}
	U(T, \bp) 
		& = -\log\left( P(D \; | \; T, \bp ) P(T, \bp) \right) \nonumber \\
		& = -\sum_j 
				\log\left(\Mult(\bx_j \, | \, \bp_j ) \right) +
				\log\left(\Mult(\by_j \, | \, T \bp_j ) \right) \nonumber \\
		& \qquad - \sum_{i=1}^d \log\left(\Dir( \bT_{\star i}\, | \, \balpha_T ) \right) 
			- \sum_j \log\left(\Dir( \bp_j \, | \, \balpha_p ) \right).
\label{eqn:iplaUdef}
\end{align}
Here $\bT_{\star i}$ is the $i$-th column of $T$ and the shape parameters for
the priors are
\begin{equation}
	\balpha_p = \underbrace{(0.5, \dots, 0.5)}_{\mbox{$d$ times}}
	\qquad \mbox{and} \qquad
	\balpha_T = \underbrace{(1.1, \dots, 1.1)}_{\mbox{$d$ times}}.
\label{eqn:shapeParamsForPriors}
\end{equation}
This choice of $\balpha_p$ corresponds to a distribution that favours $\bp_j$ in
which most of the entries are small,
while the Dirichlet distribution with $\balpha_T$ favours weakly $T_{*j}$ with uniform probabilities. 

Unfortunately one cannot use the SDEs in Eqn.~\eqref{eqn:AkyilidizSDEs} 
directly for either the columns $\bT_{\star i}$
or the vectors of probabilities of note selection $\bp_j$. The issue is that these quantities must lie in the $d$-simplex, but the SDEs cannot enforce such constraints. Thus to apply the IPLA we must transform the $\bp_j$ and columns of $T$ onto  unconstrained scales: we use \emph{logit-transformed, stick-breaking coordinates}, which are explained below.

\subsection{Unconstrained coordinates}
\label{sec:stickBreaking}

The problem of needing parameters to lie in a simplex when using an algorithm that
requires unconstrained variables is well known and we adopted the approach used by the developers of the HMC package Stan \parencite{Carpenter:2017rm}. Alternative coordinate systems exist, (e.g. softmax coordinates or the system described in \cite{Betancourt:2012qd}), but as we use Stan in Section~\ref{sec:evidence}, we chose
to minimise the number of coordinate systems we needed to introduce.

\begin{figure}
	\begin{center}	
		\includegraphics[width=10cm]{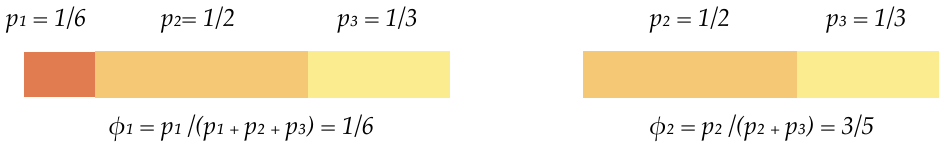}
	\end{center}
	\caption{ 
		Stick-breaking coordinates for the point 
		$p = (\frac{1}{6}, \frac{1}{2}, \frac{1}{3}) \in \Delta_3$.
		Figure S1 in the Supplementary Material illustrates all 
		the coordinate systems used in this paper.
	}
\label{fig:stickBreakingExample}
\end{figure}

The first thing to note is that the $d$-simplex is a $(d-1)$-dimensional object as, given the first $(d-1)$ of the $p_j$, we can solve the constraint
$\sum_{j=1}^d p_j = 1$ for the remaining component $p_d$. A point $\bp = (p_1, \dots, p_d)$ in the $d$-simplex $\Delta_d$ will thus be associated with $d-1$ unconstrained coordinates. We begin by constructing \emph{stick-breaking coordinates} 
$\bphi = (\phi_1, \dots, \phi_{d-1}) \in [0,1]^{d-1}$. 
To define these coordinates---and to understand the origin of their name---imagine starting with a stick of unit length and
snapping-off a piece of length $p_1$, then one of length $p_2\; \dots$ and so on. The stick-breaking coordinate $\phi_j$ is then defined as the
fraction of the remaining stick that gets snapped-off at step $j$ in this process.
Figure~\ref{fig:stickBreakingExample} illustrates the construction for an example with $\bp \in \Delta_3$, but in general we have
\begin{equation}
	\phi_1	= \frac{p_1}{p_1 + \dots + p_d}, \; \; 
	\phi_2	= \frac{p_2}{p_2 + \dots + p_d}, \; \;\dots ,\; \;
	\phi_j	= \frac{p_j}{\sum_{i=j}^d p_i},
\label{eqn:stickBreakingPToPhiToo}
\end{equation}
where the denominator of $\phi_j$ is $p_j + \dots + p_d$, the length of the stick
that remains after the first  $j-1$ pieces have been broken off.
In what follows it will prove convenient to rewrite the denominators in terms of 
the pieces already removed, so that the formulae become:  
\begin{equation}
	\phi_1	= \frac{p_1}{1}, \; \; 
	\phi_2	= \frac{p_2}{1 - p_1}, \; \;\dots ,\; \;
	\phi_j	= \frac{p_j}{1 - \sum_{i=1}^{j-1} p_i}.
\label{eqn:stickBreakingPToPhi}
\end{equation}

Next we invert the transformation to give 

\begin{align}
	p_1	&= \phi_1, \; \; 
	p_2	= \phi_2 (1 - \phi_1), \; \;\dots ,\; \;
	p_j	= \phi_j \prod_{i=1}^{j-1} (1 - \phi_i), \; \;\dots \nonumber \\
	p_d & = \prod_{i=1}^{d-1} (1 - \phi_i).
\label{eqn:stickBreakingPhiToP}
\end{align}

Although the sum of the stick-breaking coordinates is no longer constrained, the
$\phi_j$ are still restricted to satisfy $0 \leq \phi_j \leq 1$ and so are 
unsuited to use in the IPLA. We thus apply the logit transformation to the components of $\bphi$ to obtain $\bu = (u_1, \dots, u_{d-1}) \in  \R^{d-1}$ defined by
\begin{equation}
	u_j = \logit( \phi_j) = \log\left( \frac{\phi_j}{1 - \phi_j}\right),
\label{eqn:logitTransform}
\end{equation}
a transformation that's also simple to invert:
\begin{equation}
	\phi_j = \frac{e^{u_j}}{1 + e^{u_j}}.
\label{eqn:invLogitTransform}
\end{equation}
Table~\ref{tbl:coordinateChanges} summarises the coordinate changes
required to convert between $\bp \in \Delta_d$ and the unconstrained variables 
needed for the IPLA. In the remainder of the article, we will use the following
notation:

\begin{itemize}
	\item $\Phi:\Delta_d \rightarrow [0,1]^{d-1}$ will be the map that sends
	a point $\bp \in \Delta_d$ to its associated stick-breaking coordinates,
	as given by Eqn.~\eqref{eqn:stickBreakingPToPhi}.
	
	\item $\Lambda:[0,1]^{d-1} \rightarrow \R^{d-1}$ will be the map that sends
	stick-breaking coordinates to their logit-transformed counterparts, 
	as given by Eqn.~\eqref{eqn:logitTransform}.	

	\item $\bu_j = \Lambda(\Phi(\bp_j))$ will be the $(d-1)$-dimensional vector of 
	logit-transformed, stick-breaking
	coordinates associated with the tutor's note-usage probabilities $\bp_j$;

	\item $\btau_k = \Lambda(\Phi(\bT_{\star k}))$ will be the $(d-1)$-dimensional 
	set of logit-transformed, stick-breaking coordinates associated with the 
	$k$-th column of the transmission matrix.
\end{itemize}

\begin{table}[hb]
	\begin{center}
		\renewcommand{\arraystretch}{1.5}
		\begin{tabular}{ll|lc}
			\multicolumn{2}{l|}{Coordinate} & Domain & Definition \\
			\hline
			\small Simplex	& $\bp = (p_1, \dots, p_d)$ & $\bp \in \Delta_d$  \\
			\small Stick-breaking & $\bphi = (\phi_1, \dots, \phi_{d-1})$ & $0 \leq \phi_j \leq 1$ & $\phi_j = p_j/( \textstyle \sum_{i=j}^d p_i)$ \\
			\small Logit-transformed & $\bu = (u_1, \dots, u_{d-1})$ & $u_j \in \R$ & $u_j = \log\left( \phi_j/(1 - \phi_j) \right)$ \\
		\end{tabular}
	\end{center}
	\caption{ 
		Formulae for stick-breaking coordinates $\bphi$ and their logit-transformed 
		partners $\bu$ in terms of the simplex coordinates $\bp$. These changes of 
		coordinate allow us to set up the SDEs for the IPLA using unconstrained 
		variables. 
	}
	\label{tbl:coordinateChanges}
\end{table}

\section{Establishing convergence}
\label{sec:smoothness}
The proofs of convergence in \textcite{Akyildiz:2023nh} depend on three
assumptions, the first two of which involve the 
function $U(\bu, \btau)$ defined by transforming the log density in 
Eqn.~\eqref{eqn:iplaUdef} to a density over the unconstrained variables. 
In this section we introduce these assumptions and check
whether our $U(\bu, \btau)$ satisfies them. In brief, it doesn't. Although
we can arrange our implementation of the IPLA so that two of the three
conditions are satisfied, we can also prove that the remaining one is not.
In the rest of this section we give precise statements of the assumptions
and remark briefly on the extent to which our system satisfies them. 

Nevertheless, numerical
results presented in Section~\ref{sec:numericalResultsIPLA} suggest
that the IPLA does converge to a sensible MMLE for $T$ when applied 
to synthetic data and produces plausible results when applied to 
Java sparrow song. This is in keeping with the discussion in
\textcite{Akyildiz:2023nh}, where the authors say 
\begin{quote}
$\dots$ our results are obtained in the case where $U$ is gradient Lipschitz and obeys a strong convexity condition, but we believe that similar results can be obtained under much weaker (nonconvex) conditions, \textcite{Zhang:2023gd}.
\end{quote}

\subsection{Smoothness and regularity assumptions}

After suppressing terms that do not depend on the $\bp_j$ or $T$ and combining terms in $\log(p_{ij})$, the negative log-posterior function is:

\begin{align}
	U(T, \bp) 
		& = -\sum_j \sum_{i=1}^d 
				(\alpha_p - 1 + x_{ij}) \log(p_{ij})  + y_{ij} \log(q_{ij}) \nonumber \\
		& \qquad - \sum_{r=1}^d \sum_{i=1}^d 
					(\alpha_T - 1) \log\left(T_{ri} \right) \nonumber \\
		& = -\sum_j \sum_{i=1}^d 
				(\alpha_p - 1 + x_{ij}) \log(p_{ij}) 
				+ y_{ij} \log\left(\sum_{r=1}^d T_{ir} p_{rj}\right)  \nonumber \\
		& \qquad - \sum_{r=1}^d \sum_{i=1}^d (\alpha_T - 1) \log\left(T_{ri} \right) 
\label{eqn:iplaUlogTerms}
\end{align}
where the $\bp_j$ and the columns of $T$ are constrained to lie in the 
simplex $\Delta_d$.

To transform this density to one over the unconstrained coordinates $\bu_j$ and $\btau_k$, we need to include a factor arising from the Jacobian of the change of variables. Calculations detailed in Section~\ref{sec:changeOfCoords} of the Appendix establish that the Jacobian factor is just the product 
\begin{equation}
	\left(\prod_{j}\prod_{i=1}^d p_{ij}\right) \times
	\left( \prod_{r=1}^d \prod_{i=1}^d T_{ri} \right)
\label{eqn:iplaDensityJacobinaFactor}
\end{equation}
and so the optimisation target for the IPLA becomes
\begin{align}
	U(\bu, \btau) 
		& = -\sum_j \sum_{i=1}^d 
				(\alpha_p + x_{ij}) \log(p_{ij}) 
				+ y_{ij} \log\left(\sum_{r=1}^d T_{ir} p_{rj}\right)  \nonumber \\
		& \qquad - \sum_{r=1}^d \sum_{i=1}^d \alpha_T \log\left(T_{ri} \right),
\label{eqn:iplaTransformedUlogTerms}
\end{align}
where the $p_{ij}$ and the $T_{ri}$ are to be regarded as functions of the unconstrained
coordinates. 

To state the conditions, we need the following definitions:

\begin{definition*}[Lipschitz function]
	A function $f:X \rightarrow Y$ between two Banach spaces is said to
	be \textbf{Lipschitz continuous} if, for all $\bx_1, \bx_2 \in X$ we
	have, for some fixed constant $K$, that
	\begin{equation}
		\| f(\bx_1) - f(\bx_2)\|_Y \leq K\| \bx_1 - \bx_2\|_X.
	\label{eqn:Lipschitz}
	\end{equation}
	Any such $K$ is called a \emph{Lipschitz constant} for $f$ and when such
	a constant exists, we say that ``$f$ is Lipschitz".
\end{definition*}

\begin{definition*}[Strong convexity]
	A differentiable function $f:\R^n \rightarrow \R$ is said to
	be \textbf{strongly convex} if, for all $\bv, \bv' \in \R^n$ we
	have, for some fixed constant $\alpha$, that
	\begin{equation}
		\left\langle \bv - \bv', \, \nabla f(\bv) - \nabla f(\bv') \right\rangle
		\geq \alpha \| \bv - \bv' \|^2,
	\label{eqn:stronglyConvex}
	\end{equation}
	where the angle braces on the left hand side denote a Euclidean inner product.
	Strong convexity means, roughly, that the growth of $f$ is at least quadratic,
	but there are a host of equivalent conditions: see, for example, 
	\textcite{Zhou:2018dg}.
\end{definition*}

The proofs of convergence in \textcite{Akyildiz:2023nh} as applied to our 
problem then depend on the following assumptions:
\begin{description}
	\item[A1:] If the count data are held fixed, then $\nabla_{\btheta} U$ 
	is Lipschitz in $\btheta = (\bu, \btau)$ where $\bu$ encompasses
	of all the $\bu_j$. We prove this in Appendix~\ref{sec:SmoothnessProofs}.
	
	\item[A2:] The function $U(\bu, \btau)$ is strongly convex. 
	This assumption is violated for our problem and we provide a proof in 
	Appendix~\ref{sec:SmoothnessProofs}.
	
	\item[A3:] If we fix $N$, the number of particles (that is, random walkers) in 
	Eqns.~\eqref{eqn:AkyilidizSDEs} and \eqref{eqn:AkyilidizEulerMaruyama}
	of the IPLA, then the distribution from which we draw the initial condition
	should be such that 
	\begin{displaymath}
		\left( \btau_0, \bu^{1}_0/\sqrt{N}, \dots, \bu^{N}_0/\sqrt{N} \right)
	\end{displaymath}
	has bounded second moment. We have arranged for this to be true, as is 
	also shown in Appendix~\ref{sec:SmoothnessProofs}.
\end{description}

\section{Numerical results}
\label{sec:numericalResultsIPLA}
Here we report briefly on two numerical experiments, one with synthetic data 
generated by the model in Section~\ref{sec:generativeModel} and another with
the birdsong data from \textcite{Lewis:2021hj}. 

\subsection{Synthetic data}
\label{sec:syntheticData}
We generated synthetic count data as follows:
\begin{itemize}
	\item Fixed a number of notes $d = 5$, a number of aligned positions $n = 500$ and
	a number of recordings per bird $m = 8$. The values of $d$ and $m$ are similar to
	those in our real data, but we chose $n$ rather smaller, so the computations would 
	finish more quickly.
	
	\item Generated a $d \times d$ transmission matrix $T$ whose columns were drawn
	from a Dirichlet distribution with shape parameters $\alpha_i = \alpha_T = 1.0$
	for $i \in \{1, \dots, 5\}$. This corresponds to a uniform distribution over 
	$\Delta_5$.
	
	\item Generated note-usage probabilities for the tutors $\bp_j \in \Delta_5$ for 
	$j \in \{1, \dots, 500\}$ by drawing them from a Dirichlet distribution with 
	shape parameters $\alpha_i = \alpha_p = 0.5$. This favours $\bp_j$ similar to 
	observed data in that some of the entries are much, much smaller than others
	(see Section~\ref{sec:realData} for an overview of the real data).
	
	\item Generate note-usage counts $\bx_{j}$ for the tutors by drawing 
	$(x_{1j}, \dots, x_{5j})$ from the multinomial distribution with $m=8$ total
	counts and probabilities given by $\bp_j$.

	\item Use $T$ and the $\bp_j$ to compute note-usage probabilities for the pupils
	using the model $\bq_j = T \bp_j$.
	
	\item Generate note-usage counts $y_{ij}$ for the pupils by drawing from
	the multinomial distribution as sketched above.

\end{itemize}
The virtue of this approach is that we know the true value of $T$. 
Figure~\ref{fig:syntheticIPLA} summarises a numerical experiment to estimate
$T^\star$ using the IPLA applied to count data generated as above and suggests that 
the estimates, though sometimes biased, are reasonably accurate, despite the 
lack of strong convexity in the negative log posterior $U(\bu, \btau)$.

We also investigated this example for the possibility of multimodality in
the marginal likelihood by repeating the IPLA estimation of $T^\star$ with
$N = 512$ particles 200 times, initialising the algorithm differently each
time. We thus obtained, for each of the $d \times d$ entries $T^{\star}_{ij}$ in
the transmission matrix, a collection of 200 estimates and we examined the
distribution of these estimates, both visually (histograms are available in
Figure~S2 of the Supplementary Material) and with the dip test of unimodality developed
in \cite{Hartigan:1985yn}. This is a hypothesis test whose null hypothesis
is that the data are drawn from a distribution that has a continuous,
unimodal density. In all cases, the estimates for the  $T^{\star}_{ij}$
appeared sharply concentrated near a single value and the dip test 
suggested little evidence of deviation from unimodality ($p > 0.7$ in all cases).

\begin{figure}
	\begin{center} 
		% See https://tex.stackexchange.com/questions/16942/difference-between-textwidth-linewidth-and-hsize
		\includegraphics[width=\columnwidth]{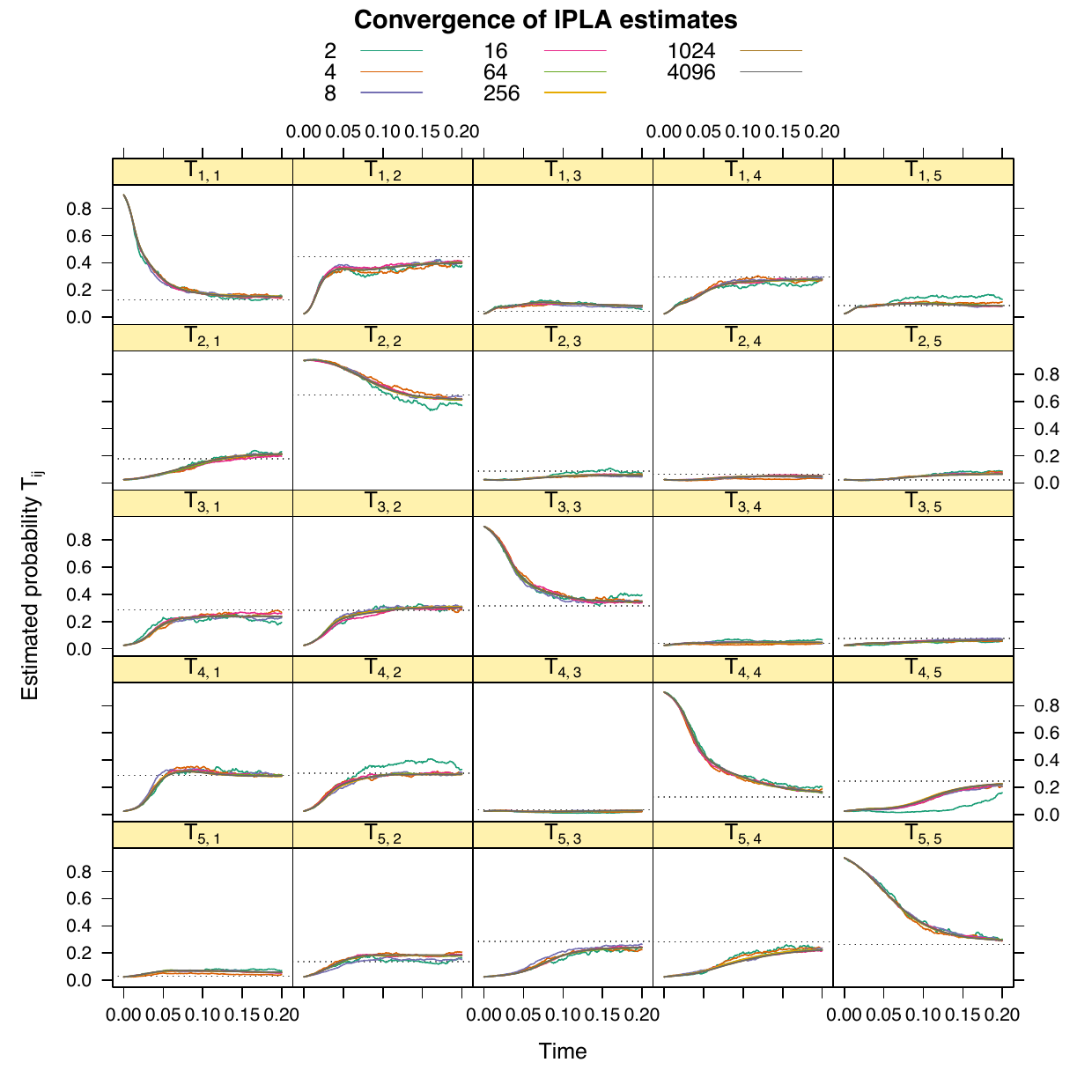}
	\end{center}
	\caption{ 
		Estimates of the elements $T^{\star}_{ij}$ of the transmission matrix 
		produced by the IPLA applied to synthetic data with
		$N \in \{2, 4, 8, 16, 64, 256, 1025, 4096\}$. Each panel shows the 
		time-course of the estimates for a single matrix element with 
		separate curves for each value of $N$ (see key at top). 
		Dashed horizontal lines indicate true values.
	}
	\label{fig:syntheticIPLA}
\end{figure}

\subsection{Data from recorded birdsong}
\label{sec:realData}
Building on work reported by \textcite{Lewis:2021hj}, we aligned note
sequences derived from separately-recorded songs for pupil-tutor pairs and
focused attention on positions where fewer than half of each bird's aligned
songs contained gap characters. Although \cite{Lewis:2021hj} defined 16
note types, many of them are sung very infrequently (see
\Cref{fig:noteUsageHeatmap} and Table~S1 in the Supplementary Material),
thus we amalgamated the eight rarest types into a single ``Other" category
to produce a total of nine note types. Next we used the IPLA to estimate a note
transmission matrix $T^\star$, producing results illustrated in
\Cref{fig:birdsongIPLA} and \Cref{fig:birdsongTransMat}. For the most
commonly-sung notes, $T^\star$ is very close to the identity, reflecting the
faithfulness of song transmission: pupils typically learn to reproduce
their tutors' notes accurately. For the three rarest note types, including the
amalgameted note class 9, pupils appear sometimes to replace the tutor's note 
with a more commonly sung one.

\begin{figure}
	\begin{center} 
		\includegraphics[width=0.85\columnwidth]{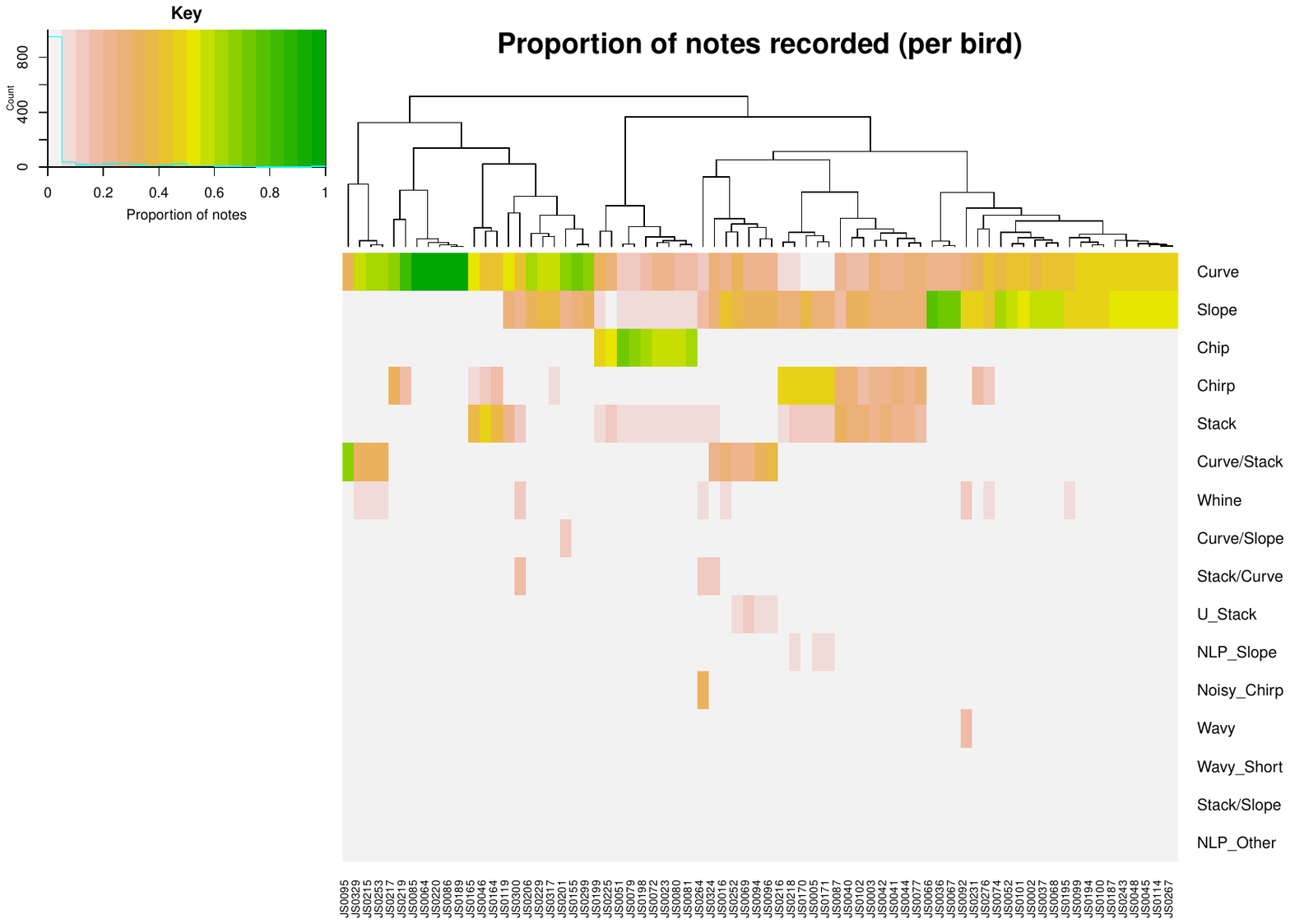}
	\end{center}
	\caption{ 
		A heatmap of note-usage proportions. Columns correspond to individual 
		birds, while rows correspond to note types. Note types are listed in order
		of decreasing frequency-of-use, so Curve notes are the most commonly-sung,
		followed by Slope and Chip notes, in that order.
		Pixels in the heatmap are coloured according to the
		proportion of a given bird's notes that were of a given
		type. The eight most commonly used note types account for more than 97.7\% of
		all notes recorded. Further details of the distribution of note
		types are available in Table~S1 in the Supplementary Material.
	}
	\label{fig:noteUsageHeatmap}
\end{figure}

\begin{figure}
	\begin{center} 
		\includegraphics[width=\columnwidth]{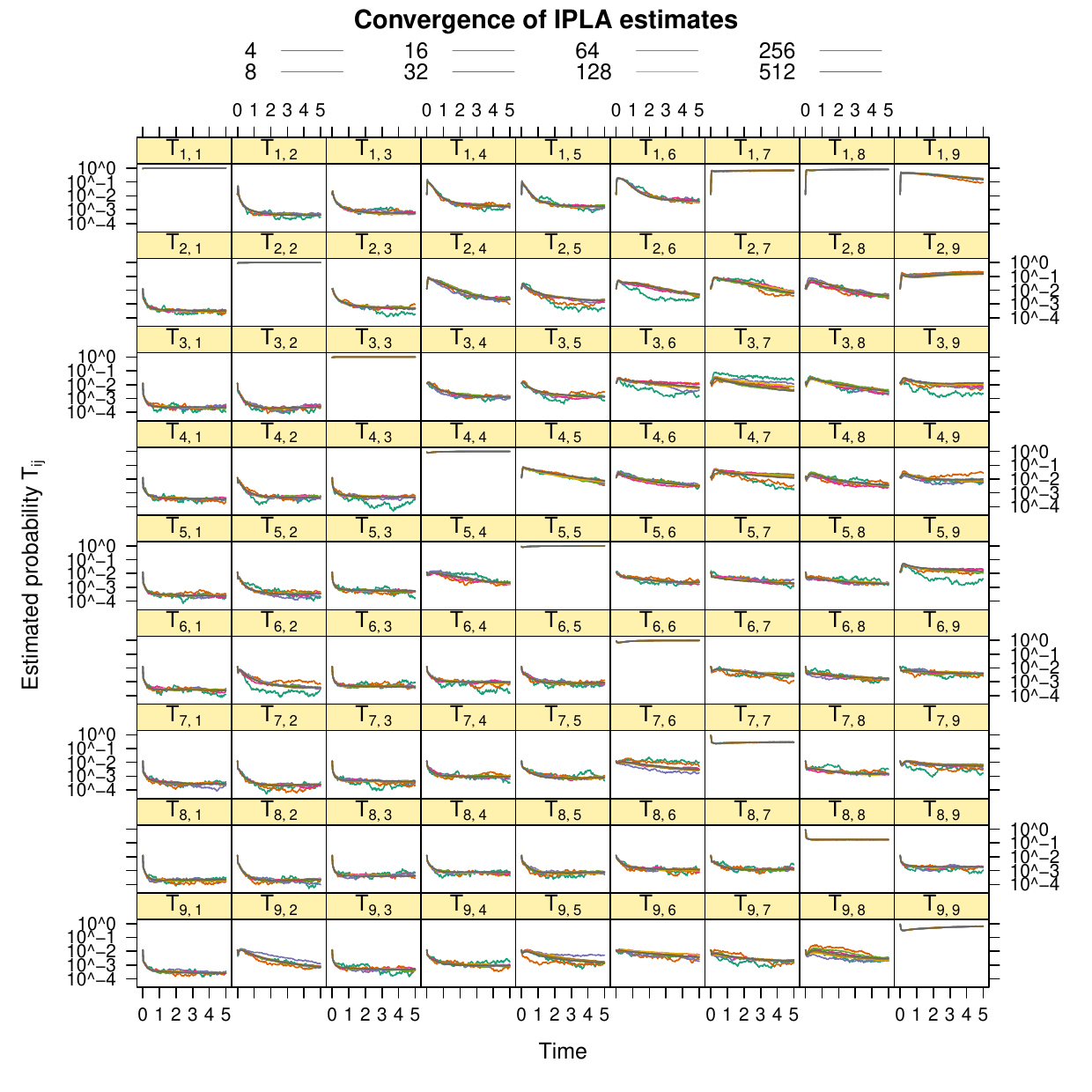}
	\end{center}
	\caption{ 
		Estimates of the elements $T^{\star}_{ij}$ produced by applying the IPLA
		to a nine-note version of the data from 
		\protect\textcite{Lewis:2021hj} with  
		$N \in \{4, 8, 16, 32, 64, 128, 256, 512\}$. Each panel shows the 
		time-course of the estimates for a single matrix element with 
		separate curves for each value of $N$ (see key at top). Note
		that the matrix elements plotted are on a log scale and the 
		majority are very small.
	}
	\label{fig:birdsongIPLA}
\end{figure}

\begin{figure}
	\begin{center} 
		% See https://tex.stackexchange.com/questions/16942/difference-between-textwidth-linewidth-and-hsize
		\includegraphics[width=\columnwidth]{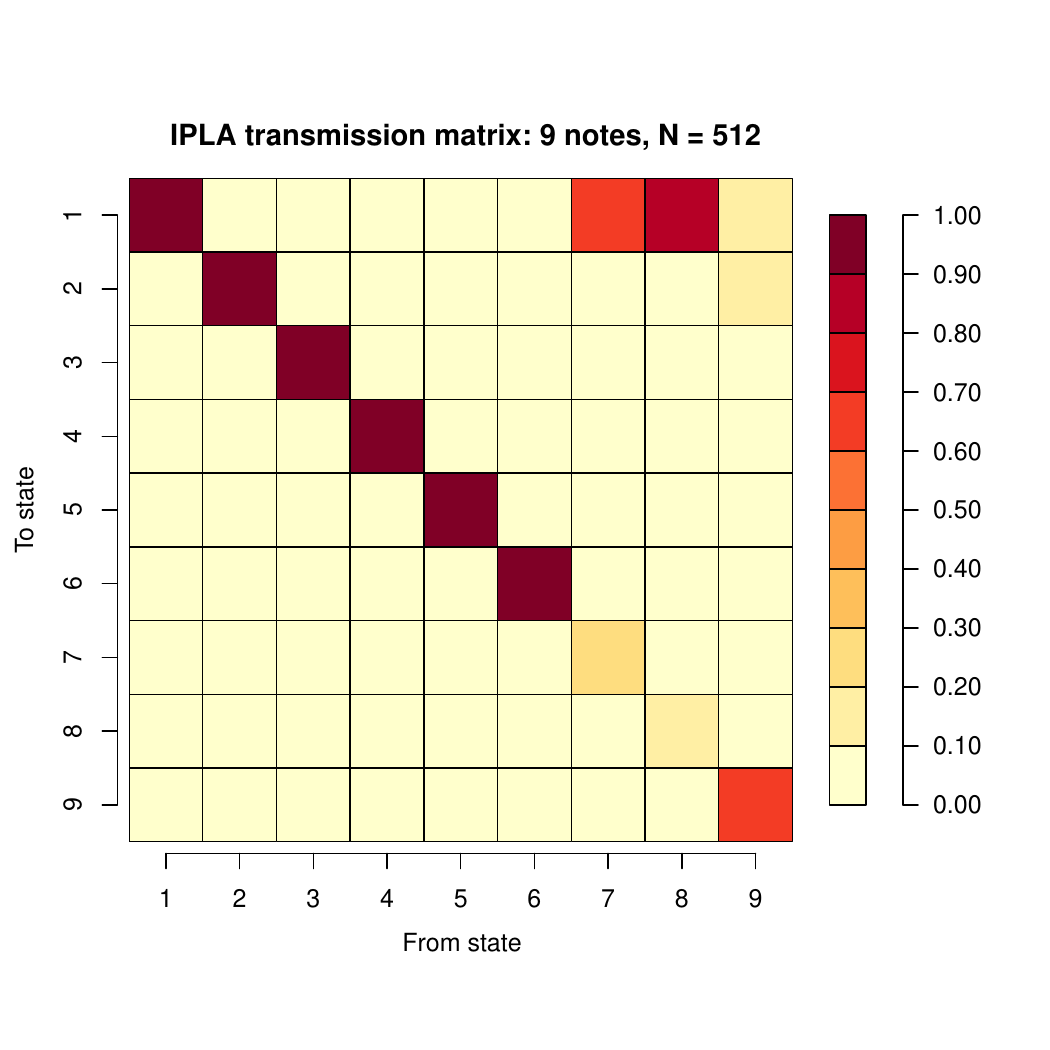}
	\end{center}
	\caption{
		A heatmap illustrating the sizes of the matrix elements 
		$T^{\star}_{ij}$ produced by 
		applying the IPLA to a nine-note version of the data from 
		\protect\textcite{Lewis:2021hj} with $N = 512$.
	}
	\label{fig:birdsongTransMat}
\end{figure}

\subsection{Transmission Matrix}
\label{sec:trans_mat}

We fitted a transmission matrix for the eight most commonly used notes,
along with a ninth, amalgamated note-type for the most rarely-sung notes,
using IPLA (\Cref{fig:birdsongTransMat}). $T^{\star}_{ij}$ is the estimated
probability that a pupil sings note $i$, given the tutor sings note $j$ in
a given position. The diagonal elements corresponding to the six most
commonly-sung notes are very close to one, suggesting that the learning
fidelity is high overall. The exceptions are notes 7, 8 and 9, which have
probabilities of 0.68, 0.8 and 0.3, respectively, of changing to one of the
most commonly-sung notes, 1 or 2. Thus, we expect only small changes in the
note sequences per generation, unless the tutor often sings notes of type
7, 8 or one of the rare notes aggregated into class 9. 

\subsection{Evidence}
\label{sec:evidence}

\begin{figure}
    \centering  \includegraphics[width=\linewidth]{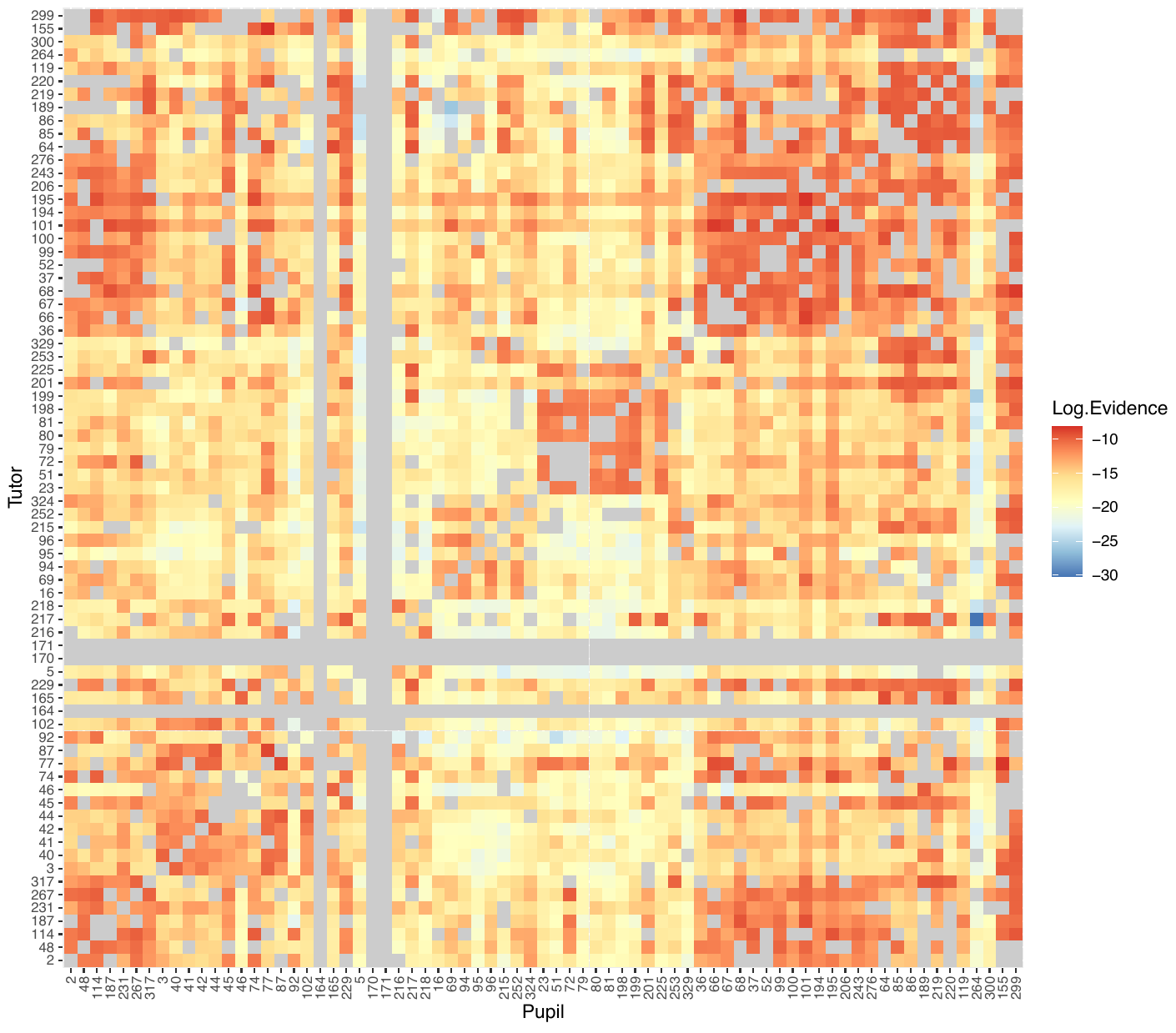}
    \caption{Heatmap of log evidence per site estimated via bridge sampling. The bird ID’s are ordered first by lineage and then, within the same lineage, numerically. The pixel for a given  (pupil, tutor) pair corresponds to the log evidence per site of a pupil-tutor relationship based on all positions in an alignment of the two birds' songs. Grey pixels indicate cases where the log evidence could not be estimated, either because the birds' songs could not be aligned or because the associated alignments had too few aligned positions, given our threshold on song count. The diagonal is grey because birds cannot tutor themselves.}
    \label{fig:logevidence_mat}
\end{figure}

We used bridge sampling (\cite{Gelman:1998ta, Gronau:2017pd}) as implemented in the
\texttt{bridgesampling} R package of \cite{Gronau:2020fb}, in combination with 
the HMC sampler Stan (\cite{Carpenter:2017rm}) to estimate the probability that any one bird tutored another. That is, for each potential (pupil, tutor) pair $(a,b)$,  
we aligned the birds' songs and then estimated
\begin{align}
	E_{ab} 
		& =  \prod_j  \int P(\bx_j \, | \, \bp_j ) P(\by_j \, | \, \bp_j, T^\star ) 
			P(\bp_j) \, d\bp_j \nonumber \\
		& = \prod_j  \int \Mult(\bx_j \, | \, \bp_j ) \Mult(\by_j \, | \, T^
			\star \bp_j )\, \Dir( \bp_j \, | \, \balpha_p ) \, d\bp_j
\label{eqn:singlePairEvidence}
\end{align}
where here the index $j$ ranges over all aligned positions for the putative 
pupil-tutor pair, $\bx_j$ is the putative tutor's note usage counts at position $j$,
$\by_j$ are those of the putative pupil and $\balpha_p$ is as in Eqn.~\eqref{eqn:shapeParamsForPriors}. 
We will refer to this quantity 
as the \emph{evidence} for the pair and,
as it depends on the number $n_{ab}$ of aligned positions, will work with a related quantity
\begin{equation}
	e_{ab} = \frac{1}{n_{ab}} \log( E_{ab} ), 
\label{eqn:logEvidencePerSite}
\end{equation}
the mean log evidence per aligned site, to permit comparisons between different
putative pupil-tutor pairs.

In addition to bridge sampling, we also estimated the evidence with a harmonic mean estimator (see \cite{Robert:2009kc} and \cite{Gronau:2017pd}).
The results typically agreed well with those from bridge sampling (differences of less that 2\% in the estimates of $e_{ab}$ for 94\% of putative pairs), but were more computationally expensive to obtain and had a larger variance than those from bridge sampling, so we haven't reported them. 

\Cref{fig:logevidence_mat} shows log evidence per aligned site as specified by Eqns.~\eqref{eqn:singlePairEvidence} and \eqref{eqn:logEvidencePerSite}. The blocks of high log evidence near the diagonal align broadly  with the social lineages, suggesting that our method can pick up some social relationships among birds. However, the considerable variation within these blocks suggest that this is not always the case. Furthermore, the evident symmetry of the matrix suggests that our method may struggle to identify the direction of transmission. We were unable to compute the evidence for some pairs, either because we were unable to align their songs or because there were insufficiently many aligned positions. This occurs when the song sequences are too dissimilar (e.g. birds sing completely different notes most of the time).  

\begin{figure}
    \centering   
    \includegraphics[width=0.85\linewidth]{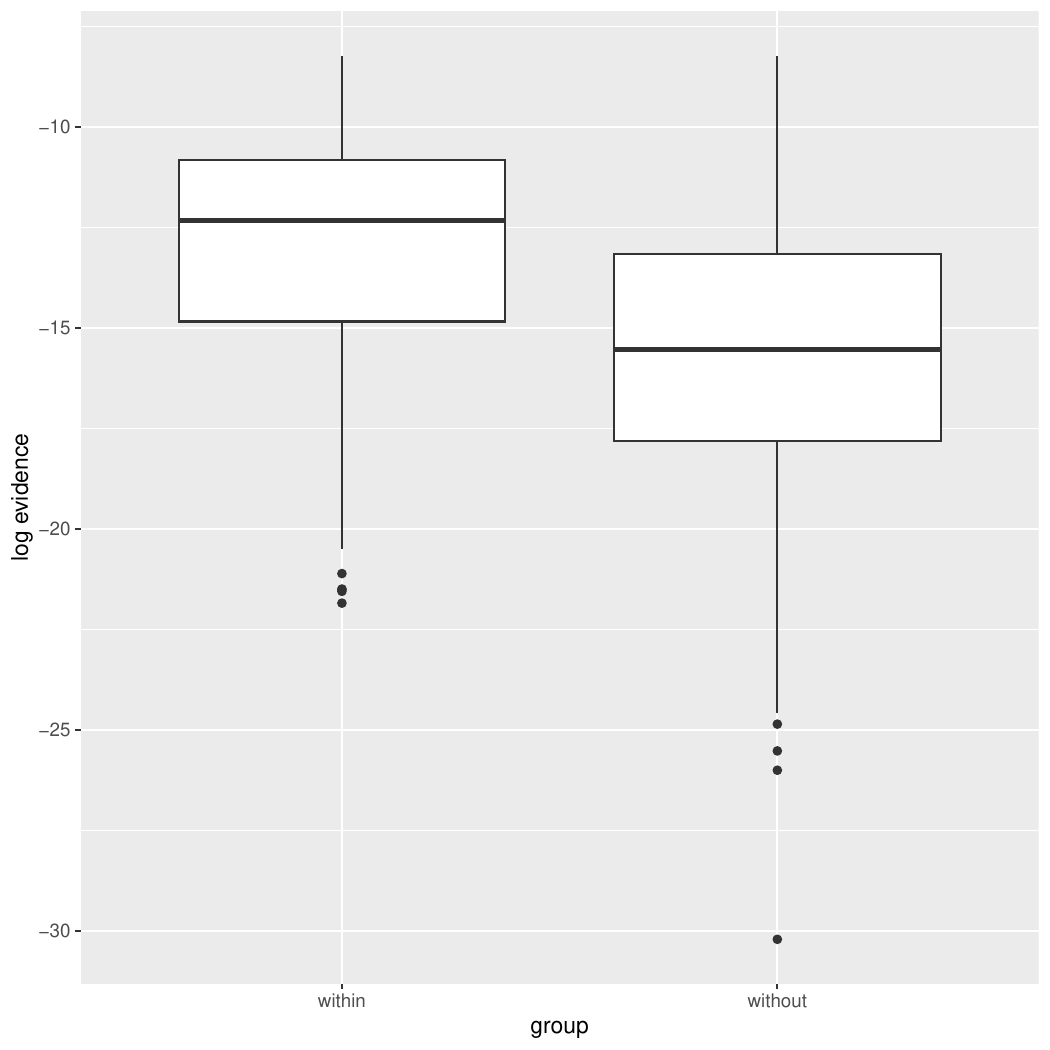}
    \caption{
    	A boxplot of log evidence values for putative pupil-tutor
    	relationships, comparing values for pairs from within the same
    	lineage with those for pairs from different lineages. Same lineage
    	pairs have a higher median log evidence, but there is some overlap
    	in the two interquartile ranges.
    }
    \label{fig:logevi-lineage-bp}
\end{figure}

First, we considered whether our method can accurately discern the direction of song tutoring within the 57 pupil-tutor pairings provided in the original data. We compared the log evidence for the true pairings with those in which the tutor and pupil swapped roles and performed a paired Wilcoxon signed-ranked test (0.05 significance) to determine whether the true pairings had a significantly higher log evidence than the reverse. Our results were insignificant ($p$-value 0.3217, V = 581), thus we could not correctly determine the direction of song tutoring. This is consistent with our informal observation of the symmetry of \Cref{fig:logevidence_mat}. 

Next, we considered whether we could use the log evidence results to deduce pupil-tutor pairings, regardless of the direction. For every pupil in the original dataset, we took our predicted tutor to be the bird with the highest log evidence. We compared this against 57 known pupil-tutor pairs and found the method to be correct only 2 times. This is also consistent with \Cref{fig:logevidence_mat} where the columns, representing the pupils, do not typically have a single pixel showing substantially higher log evidence than the rest. Instead, there are several pixels showing high evidence, indicating several potential song tutors. 

Moving beyond individual pairings, we investigated whether our method could distinguish between first and second generation pupil-tutor relationships. For all cases where the comparison was possible, we compared the log evidence for the pupils and their tutors with that for the pupils and their grand tutors (tutor of their tutor). As above, we compared the two sets of values using a paired Wilcoxon signed-ranked test (0.05 significance). The results were insignificant (V = 166, $p$-value = 0.2221). 

Finally, we investigated whether social lineages could be detected by our method. First, we classified each putative pupil-tutor pair according to whether both birds were part of the same lineage or not. We then compared the two sets of log evidence values---those from within-lineage pairs and those from cross-lineage pairs---graphically: see \Cref{fig:logevi-lineage-bp}. Although the log evidence tended to be higher for the within-lineage pairs, there was some overlap in the inter-quartile ranges and outliers in both sets.  We also tried to use log evidence to predict a bird's lineage by considering all putative (tutor, pupil) pairs in which the given bird played a role, associating the pair with the lineage of the given bird's partner. We then assigned the bird to the lineage whose pairs had the highest mean log evidence.

 We noticed that the two Dark Pink birds (JS299, JS155) had very high log evidence values with most other birds, excluding the cases where the evidence could not be estimated. These two birds to sing mostly A's, interspersed with some B's. Since A is the most commonly sung note, this could explain why the birds gave unusually high log evidence values. When we remove these two birds from the analysis, then out of the 70 remaining birds, the approach sketched above assigned 39 to the correct lineage. This corresponds to a predictive accuracy of $\approx 55 \%$, whereas a naive classifier using the most common class would have $\approx 21 \%.$ We also evaluated our predictions using the Adjusted Rand Index (ARI), which measures the similarity between two clusterings while accounting for chance agreement. Unlike classification accuracy, ARI captures the overall grouping structure rather than point-by-point label matches. Our predictions achieved an ARI of 0.2763, indicating performance better than random labelling (ARI = 0) but well short of a perfect match (ARI = 1), consistent with the observed classification accuracy of 55\%.
Both the ARI and percentage accuracy are consistent with \Cref{fig:logevidence_mat}, where the blocks along the diagonal, which align broadly with the lineages, are not regions of consistently high log evidence. Thus our model made modestly accurate predictions of social lineage. This is consistent with \Cref{fig:logevidence_mat}, where the blocks along the diagonal, which align broadly with the lineages, are not regions of consistently high log evidence. 

\section{Conclusion}
\label{sec:conclusion}

Using Java sparrows as a model organism, we have presented the first sequence evolution model of birdsong transmission. By applying the IPLA of \cite{Akyildiz:2023nh}, we estimated a note transmission matrix for the entire population. We regard this as a first approximation: the variation in note-usage pattern evident in \Cref{fig:noteUsageHeatmap} might suggest the use of a Bayesian hierarchical model in which the prior for the note usage probabilities varies from lineage to lineage. As a further simplification, we lumped the eight rarest note-types into a single catch-all category, numbered 9 in our scheme. 

Generally, our estimated transmission matrix is consistent with the well-established faithfulness of song transmission: most diagonal elements $T^{\star}_{jj}$ are very close to 1.0. The three exceptions include the two comparatively rarely-sung note classes seven and eight---Whine and Curve/Slope, respectively---which pupils sometimes convert to the most commonly-sung note type, Curve. The catch-all class of very rare notes, nine, shows similar conversions to the two most commonly-sung note classes. This may be an artifact of the alignment process, during which phrases containing rarely-sung notes may be aligned with similar phrases that lack these notes. Ideally, one would incorporate the probabilistic models underpinning alignment along with those representing transmission into one grand estimation problem, but this would be computationally intractable.

Using the estimated note transmission matrix, we computed bridge sampling estimates for the log evidence for all potential pupil-tutor relationships. Although our model was not able to detect pairwise relationships (e.g. correctly match pupils to tutors) or the direction of song tutoring, we did achieve modest results in assigning pupils to their correct lineages. The key issue appears to be the faithfulness of song transmission and the comparative sparsity of data, particularly the song lengths. In bioinformatic contexts, where mistakes during, for example, DNA replication are also very rare, researchers are typically working with kilobases or even megabases worth of data. For example, the median length of a human gene is 24 kilobases (\cite{Scherer:2008ph}) and the human genome is $\approx $ 3 billion base pairs long. Even with a very small mutation rate of $0.5 \times 10^{-9}$ (\cite{Scally:2016pf}), there are still a substantial number of variants to study. By contrast, our songs are all less than 90 notes long and their collections of aligned phrases are even shorter, making observed note changes far rarer. 

Furthermore, our sequence evolution model lies downstream of the preprocessing of the song recordings conducted by \cite{Lewis:2023gp,Lewis:2021hj}. There are two main preprocessing steps; namely the segmentation of notes from the song spectrograms and the classification of notes into note classes. The recordings were taken in controlled conditions with minimal background noise, so it was relatively simple to manually identify the start and end points of notes. After segmentation, \cite{Lewis:2021hj} manually classified the notes using a suite of acoustic characteristics. They repeated the classification with a second observer, achieving a consensus of $97.5 \%.$ Although our model is dependent on the work of \cite{Lewis:2021hj}, we are confident that this was done as accurately as anything in the field. Modern tools such as TweetyNet (\cite{Cohen:2022xn}), which uses a neural network to preprocess song recordings from a range of birds could increase the pool of available data substantially and expand the scope for sequence evolution models to other species of songbirds. 

Additionally, our sequence evolution model relies on the multiple sequence alignments of the pupil and tutor songs. We produced these alignments using profile hidden Markov models, which treat transmission at all sites as independent of each other. Arguably this assumption doesn't hold as birds' songs are known to contain motifs. For example, if a bird sings ``ACACAC'', then the probability of singing a C depends on whether the bird has sung an A immediately before. The same challenge is found in bioinformatics where genetic data may contain a variety of motifs and repeated elements. Despite this problem, researchers have still been able to derive useful biological insights from their alignments (\cite{Reyes:2017nu, SkewesCox:2014rp}). In our case, we inspected the alignments to check that they identified regions of consensus: see \Cref{fig:alignment} for an example. Whilst relaxing the independent sites assumption may be more accurate, it would involve formulating a new, complex approach to multiple sequence alignment that would be beyond the scope of this paper. 

To establish that our approach would perform better given enough aligned motifs, we could in principle repeat our analysis with synthetic data. Currently the median number of aligned positions per pair is 18, but, using the profile hidden Markov models
of \cite{Kwong:2025dv}, we could in principle generate songs with arbitrarily long motifs, but the computations reported here already lie near the limit of our resources. 

\noindent \textbf{Acknowledgements}  \newline \noindent
The authors acknowledge the assistance provided by the Research IT team and the use of the Computational Shared Facility at The University of Manchester. 
We also acknowledge Dr.~R.~Tucker~Gilman who was Kwong's co-supervisor. 
Finally, we thank Dr.~Rebecca~N.~Lewis, Professor~Masayo~Soma (University of Hokkaido) and the members of her lab, for giving us access to their data and sharing their expertise.

\noindent \textbf{Funding} \newline \noindent
Kwong was supported by EPSRC Doctoral Training Programme grant EP/W524347/1.

\noindent \textbf{Supplementary Material}  \newline \noindent
A schematic diagram illustrating the coordinate systems used in the paper;
a set of histograms illustrating the lack of multimodality in the IPLA
estimates of the entries $T^{*}_{ij}$ in the transition matrix from
Section~\ref{sec:syntheticData} and a table listing the note types
appearing in our data along with the frequencies with which they are sung.

\newpage
\appendix
\section{Smoothness and regularity conditions: proofs}
\label{sec:SmoothnessProofs}
In this section we prove the following:

\begin{prop} \label{prop:transformedDensity}
	When one transforms the negative log posterior density whose $\bp$ and $T$-dependent
	terms are listed in Eqn.~\eqref{eqn:iplaUlogTerms} to the unconstrained coordinates,
	the terms that depend on $\bu$ and $\btau$ are
	\begin{align}
		U(\bu, \btau) 
			& = -\sum_j \sum_{i=1}^d 
				(\alpha_p + x_{ij}) 
					\log\left(\Phi^{-1}_i(\Lambda^{-1}(\bu_j))\right) \nonumber \\
			& \qquad -\sum_j \sum_{i=1}^d y_{ij} 
				\log\left(\sum_{k=1}^d \Phi^{-1}_i(\Lambda^{-1}(\btau_k))
					\Phi^{-1}_k(\Lambda^{-1}(\bu_j))\right) \nonumber \\
			& \qquad - \sum_{k=1}^d \sum_{i=1}^d 
				\alpha_T \log\left(\Phi^{-1}_i(\Lambda^{-1}(\btau_k))\right)
		\label{eqn:iplaUlogTermsUnconstrained}
	\end{align}
	where $\Lambda^{-1}$ and $\Phi^{-1}$ are the inverses of the functions 
	$\Lambda$ and $\Phi$ defined in Table~\ref{tbl:coordinateChanges}, 
	$\Phi^{-1}_i(\bphi)$ is the $i$-th component of $\Phi^{-1}(\bphi)$, 
	$\bu_j \in \R^{d-1}$
	is the vector of unconstrained coordinates associated with the tutor's
	note-usage probabilities at alignment position $j$ and $\btau_k$ is the
	vector of unconstrained coordinates associated with the $k$-th column
	of the transition matrix.
\end{prop}

\begin{theorem} \label{thm:A1}
	The gradient of $U(\bu,\btau)$ is Lipschitz and so 
	assumption A1 of Akylidiz et al. is satisfied.
\end{theorem}

\begin{theorem} \label{thm:A2}
	$U(\bu,\btau)$ is not strongly convex and so assumption 
	A2 of Akylidiz et al. is violated.
\end{theorem}

\begin{theorem} \label{thm:A3}
	Let $\tau_{ij,0}$ denote the initial condition for the unconstrained 
	coordinate $\tau_{ij}$ in the IPLA and let
	$u_{ij,0,k}$ indicate the initial condition for the unconstrained coordinate
	$u_{ij}$ for the $k$-th particle. The distributions from which these
	initial conditions are drawn are such that there exists a constant $H$
	such that, for all $k$,
	\begin{displaymath}
		\sum_{i,j = 1}^d \E[\tau_{ij,0}^2] + 
		\frac{1}{N} \sum_j \sum_{i=1}^d \E[u_{ij,0,k}^2] \leq H
	\end{displaymath}
	and so assumption A3 of Akylidiz et al. is satisfied.
\end{theorem}

\subsection{Preliminaries: changes of coordinate}
\label{sec:changeOfCoords}
We will need to compute derivatives with respect to the unconstrained coordinates, so 
here we compute the Jacobians of the coordinate changes from stick-breaking to 
simplex coordinates and from logit-transformed variables to those constrained to lie in the unit interval. The first Jacobian we will need is $J_{\bphi}$,
a $d \times (d-1)$ matrix that relates stick breaking coordinates 
$\bphi = (\phi_1, \dots, \phi_{d-1})$ to points $\bp = (p_1, \dots, p_d)$ in the
$d$-simplex. It is given by
\begin{equation}
	\def\arraystretch{1.5}
	J_\bphi = \left[ \begin{array}{cccc}
		\frac{\partial p_1}{\partial \phi_1} & \frac{\partial p_1}{\partial \phi_2} 
			& \cdots & \frac{\partial p_1}{\partial \phi_{d-1}} \\ 
		\frac{\partial p_2}{\partial \phi_1} & \frac{\partial p_2}{\partial \phi_2} 
			&  &  \vdots \\ 
		\vdots & & \ddots & \vdots  \\
		\frac{\partial p_{d-1}}{\partial \phi_1} & \cdots & &
		\frac{\partial p_{d-1}}{\partial \phi_{d-1}} \\
		\frac{\partial p_d}{\partial \phi_1} & \cdots & &
		\frac{\partial p_d}{\partial \phi_{d-1}} \\
	\end{array} \right]
	\; = \; 
	\left[ \begin{array}{cccc}
		\frac{p_1}{\phi_1} & 0 & \cdots & 0 \\ 
		\frac{-p_2}{1 - \phi_1} & \frac{p_2}{\phi_1}
			&  &  \vdots \\ 
		\vdots & & \ddots & \vdots  \\
		\frac{-p_{d-1}}{1 - \phi_1} & \cdots & &
		\frac{p_{d-1}}{\phi_{d-1}} \\
		\frac{- p_d}{1 - \phi_1} & \cdots & &
		\frac{- p_d}{1 - \phi_{d-1}} 
	\end{array} \right],
\label{eqn:sbPhiToPJacobian}
\end{equation}
where 
\begin{displaymath}
	\def\arraystretch{1.5}
	\frac{\partial p_i}{\partial \phi_j} = 
	\left\{ \begin{array}{ccl}
		-\phi_i \prod_{r=1, r \neq j}^{i-1} (1 - \phi_r) = -p_i/ (1 - \phi_j) & \; & \mbox{If $j < i < d$} \\
		\prod_{r=1}^{i-1} (1 - \phi_r) = p_i/\phi_j & \; & \mbox{If $j = i$ and $i < d$ } \\
		0 &&  \mbox{If $j > i$ and $i < d$} \\
		-\prod_{r=1, r \neq j}^{d-1} (1 - \phi_r) = -p_d/ (1 - \phi_j) & \; & \mbox{If $i = d$} \\
	\end{array} \right.
\end{displaymath}
which follows from the expressions in Eqn.~\eqref{eqn:stickBreakingPhiToP}.
In addition to $J_\bphi$, we will need $J_{\bu}$,
a $(d - 1) \times (d-1)$ matrix that relates the stick-breaking coordinates 
$\bphi = (\phi_1, \dots, \phi_{d-1})$ to their logit-transformed partners
$\bu = (u_1, \dots, u_{d-1})$. As $\phi_j$ depends on $u_j$ alone, 
$J_{\bu}$ is diagonal with entries given by 
\begin{align}
	\frac{\partial \phi_j}{\partial u_j}  
		& = \frac{d}{du_j} \left(\frac{e^{u_j}}{1 + e^{u_j}} \right) \nonumber \\
		& = \frac{e^{u_j}}{1 + e^{u_j}} - \frac{e^{2u_j}}{(1 + e^{u_j})^2} \nonumber \\
		& = \frac{e^{u_j}(1 + e^{u_j})}{(1 + e^{u_j})^2} - \frac{e^{2u_j}}{(1 + e^{u_j})^2} \nonumber \\
		& = \frac{e^{u_j} + e^{2u_j} - e^{2u_j}}{(1 + e^{u_j})^2} \nonumber \\
		& = \frac{e^{u_j}}{1 + e^{u_j}}  \frac{1}{1 + e^{u_j}}\nonumber \\
		& = \phi_j (1 - \phi_j),
\label{eqn:uJacobianDiag}
\end{align}
where, to obtain the final line, we used
\begin{equation}
	1 - \phi_j \; = \; 1 - \frac{e^{u_j}}{1 + e^{u_j}} 
	\; = \; \frac{1 + e^{u_j}}{1 + e^{u_j}} - \frac{e^{u_j}}{1 + e^{u_j}}
	\; = \; \frac{1}{1 + e^{u_j}}.
\label{eqn:oneMinusPhi}
\end{equation}

In the arguments that follow, it will prove useful to have the product
\begin{align}
	J_\bphi J_\bu  & = 
		\left[ \begin{array}{cccc}
			\frac{p_1}{\phi_1} & 0 & \cdots & 0 \\ 
			\frac{-p_2}{1 - \phi_1} & \frac{p_2}{\phi_1}
				&  &  \vdots \\ 
			\vdots & & \ddots & \vdots  \\
			\frac{-p_{d-1}}{1 - \phi_1} & \cdots & &
			\frac{p_{d-1}}{\phi_{d-1}} \\
			\frac{- p_d}{1 - \phi_1} & \cdots & &
			\frac{- p_d}{1 - \phi_{d-1}} 
		\end{array} \right]
		\left[ \begin{array}{ccc}
			\phi_1 (1 - \phi_1) & \cdots & 0 \\
			\vdots & \ddots & \vdots \\
			0 & \cdots & \phi_{d-1} (1 - \phi_{d-1})
		\end{array} \right] \nonumber  \\
	& = \left[ \begin{array}{cccc}
			p_1 (1 - \phi_1) & 0 & \cdots & 0 \\ 
			-p_2 \phi_1 & p_2 (1 - \phi_2) &  &  \vdots \\ 
			\vdots & & \ddots & \vdots  \\
			-p_{d-1} \phi_1 & \cdots & & p_{d-1}( 1 - \phi_{d-1}) \\
			- p_d \phi_1 & \cdots & & - p_d \phi_{d-1}
		\end{array} \right].
\label{eqn:JphiJu}
\end{align}
The $(d-1)\times(d-1)$ block consisting of the top $d-1$ rows of the 
product $J_\bphi J_\bu $ is the Jacobian 
\begin{align}
	J_{\bp\bu} &= 	\left[ \begin{array}{cccc}
		\frac{\partial p_{1j}}{\partial u_{1j}} & \frac{\partial p_{1j}}{\partial u_{2j}} 
			& \cdots & \frac{\partial p_{1j}}{\partial u_{(d-1)j}} \\ 
		\frac{\partial p_{2}}{\partial u_{1}} & \frac{\partial p_{2j}}{\partial u_{2j}} 
			&  &  \vdots \\ 
		\vdots & & \ddots & \vdots  \\
		\frac{\partial p_{(d-1)j}}{\partial u_{1}} & \cdots & &
		\frac{\partial p_{(d-1)j}}{\partial u_{(d-1)j}} 
	\end{array} \right] \nonumber \\
	& = \left[ \begin{array}{cccc}
			p_1 (1 - \phi_1) & 0 & \cdots & 0 \\ 
			-p_2 \phi_1 & p_2 (1 - \phi_2) &  &  \vdots \\ 
			\vdots & & \ddots & \vdots  \\
			-p_{d-1} \phi_1 & \cdots & & p_{d-1}( 1 - \phi_{d-1}) 
		\end{array} \right].
\label{eqn:uToPJacobian}
\end{align}
of the transformation from $u \in \R^{d-1}$ to the tuple $(p_1, \dots, p_{d-1})$ obtained by projecting-out the first $d-1$ components of a point $p \in \Delta^d$ in the $d$-simplex. This Jacobian will play a role in the next section, about transforming densities.

\subsection{Transforming the posterior density}
\label{sec:transformedDensity}

\begin{proof}[Proof of Proposition~\ref{prop:transformedDensity}]
We wish to transform the posterior likelihood
\begin{align*}
	L(\bp, T) & = \left(
		\prod_j \Mult(\bx_j \, | \, \bp_j ) \Mult(\by_j \, | \, T \bp_j ) 
		\Dir( \bp_j \, | \, \balpha_p )
	\right) \\
	& \qquad \times \left(
		\prod_{i=1}^d \Dir( \bT_{\star i}\, | \, \balpha_T ) \right)
\end{align*}
which, up to normalisation, is the posterior density over $(\bp, T)$, to the corresponding 
unnormalised posterior density over the unconstrained variables $\bu$ and $\btau$. The usual change of variables formula for densities says that
\begin{align}
	L(\bu, \btau) 
	& = \left(
			\prod_j \Mult(\bx_j \, | \, \bp_j(\bu_j) ) \Mult(\by_j \, | \, T(\btau)
			\bp_j(\bu_j) ) \Dir( \bp_j(\bu_j) \, | \, \balpha_p )
		\right) \nonumber \\
	& \qquad \times \left(
			\prod_{i=1}^d \Dir( \bT_{\star i}(\btau_i)\, | \, \balpha_T ) 
		\right) \times \left| \det\left( \frac{\partial (\bp, T)}{\partial (\bu, \btau)}\right) \right|.
\label{eqn:transformedPosteriorWithJacobian}
\end{align}
where the final factor in the product above is the determinant of a block-diagonal Jacobian whose blocks are of two kinds
\begin{align}
	\frac{\partial \bp_j}{\partial \bu_j} & = 
	\left[ \begin{array}{cccc}
		\frac{\partial p_{1j}}{\partial u_{1j}} & \frac{\partial p_{1j}}{\partial u_{2j}} 
			& \cdots & \frac{\partial p_{1j}}{\partial u_{(d-1)j}} \\ 
		\frac{\partial p_{2}}{\partial u_{1}} & \frac{\partial p_{2j}}{\partial u_{2j}} 
			&  &  \vdots \\ 
		\vdots & & \ddots & \vdots  \\
		\frac{\partial p_{(d-1)j}}{\partial u_{1}} & \cdots & &
		\frac{\partial p_{(d-1)j}}{\partial u_{(d-1)j}} 
	\end{array} \right]
	\;\; \mbox{and} \nonumber \\
	\frac{\partial T_{\star i}}{\partial \btau_i} & =
	\left[ \begin{array}{cccc}
		\frac{\partial T_{1i}}{\partial \tau_{1i}} & \frac{\partial T_{1i}}{\partial \tau_{2i}} 
			& \cdots & \frac{\partial T_{1i}}{\partial \tau_{(d-1)i}} \\ 
		\frac{\partial T_{2}}{\partial \tau_{1}} & \frac{\partial T_{2i}}{\partial \tau_{2i}} 
			&  &  \vdots \\ 
		\vdots & & \ddots & \vdots  \\
		\frac{\partial T_{(d-1)i}}{\partial \tau_{1}} & \cdots & &
		\frac{\partial T_{(d-1)i}}{\partial \tau_{(d-1)i}} 
	\end{array} \right].
\label{eqn:JacobianBlocks}
\end{align}
Both are the Jacobians of transformations from $\R^{d-1}$ to tuples of coordinates
such as $(p_1, \dots, p_{d-1})$ with $p_i \geq 0 \; \forall i$ and 
$\sum_{i=1}^{d-1} \leq 1$ and both have a form similar to that of $J_{\bp\bu}$ 
in Eqn.~\eqref{eqn:uToPJacobian}. In particular, both sorts of blocks are lower-triangular,
and so their determinants given by the products of their diagonal elements:
\begin{align}
	\det\left(\frac{\partial \bp_j}{\partial \bu_j} \right) 
		& = \prod_{i=1}^{d-1} \frac{\partial p_{ij}}{\partial u_{ij}} 
		= \prod_{i=1}^{d-1} p_{ij}(1 - \phi_{ij}) \nonumber \\
		&= \left(\prod_{i=1}^{d-1} p_{ij} \right)
			\left( \prod_{i=1}^{d-1} (1 - \phi_{ij})\right)
		= \left(\prod_{i=1}^{d-1} p_{ij} \right) p_{dj} \nonumber \\
		&= \prod_{i=1}^{d} p_{ij} 
\label{eqn:uJacobianBlocks}
\end{align}
where $\phi_{ij} = \Lambda^{-1}(u_{ij})$ and the equality in
the second line follows from Eqn.~\ref{eqn:stickBreakingPhiToP}. Similar calculations
lead to the conclusion that 
\begin{equation}
	\det \left( \frac{\partial T_{\star i}}{\partial \btau_i} \right) =
	\prod_{r=1}^{d} T_{ri} 
\label{eqn:tauJacobianBlocks}
\end{equation}

Combining Eqns.~\ref{eqn:transformedPosteriorWithJacobian}--\ref{eqn:tauJacobianBlocks}
leads to the conclusion that the transformed posterior likelihood is given by
\begin{align}
	L(\bu, \btau) 
	& = \left(
			\prod_j \Mult(\bx_j \, | \, \bp_j(\bu_j) ) \Mult(\by_j \, | \, T(\btau)
			\bp_j(\bu_j) ) \right) \nonumber \\
	& \qquad \times \left( \Dir( \bp_j(\bu_j) \, | \, \balpha_p )
		\left(\prod_{i=1}^{d} p_{ij} \right)\right) \nonumber \\
	& \qquad \times \left(
			\prod_{i=1}^d \Dir( \bT_{\star i}(\btau_i)\, | \, \balpha_T ) 
			\left( \prod_{r=1}^{d} T_{ri}  \right)
		\right).
\label{eqn:transformedPosteriorLikelihood}
\end{align}
Taking logs and dropping terms that don't depend on the $\bu_j$ or $\btau_j$ leads to
\begin{align}
	\log\left(L(\bu, \btau) \right)
	& = \sum_j \sum_{i=1}^d x_{ij} \log(p_{ij}(\bu_j)) 
		+ y_{ij} \log(q_{ij}(\bu_j, \btau)) \nonumber \\
	& \qquad + \sum_j \sum_{i=1}^d \alpha_p \log(p_{ij}(\bu_j))  
		+ \sum_{i=1}^d \sum_{r=1}^d \alpha_T \log\left(T_{ri}(\btau_i)\right) \nonumber \\ 
	& = \sum_j \sum_{i=1}^d (\alpha_p + x_{ij}) \log(p_{ij}(\bu_j)) 
		+ y_{ij} \log(q_{ij}(\bu_j, \btau)) \nonumber \\ 
	& \qquad + \sum_{i=1}^d \sum_{r=1}^d \alpha_T \log\left(T_{ri}(\btau_i)\right).
\label{eqn:transformedPosteriorLikelihoodWithDependency}
\end{align}
where $\bq_j(\bu_j, \btau) = T(\btau) \bp_j$. Making the substitutions
\begin{align}
	U(\bu, \btau) = - \log\left(L(\bu, \btau) \right), \;\;
	p_{ij} = \Phi^{-1}_i(\Lambda^{-1}(\bu_j)), \nonumber \\
	q_{ij} = \sum_{k=1}^d \Phi^{-1}_i(\Lambda^{-1}(\btau_k))
		\Phi^{-1}_k(\Lambda^{-1}(\bu_j)) \;\; \mbox{and} \;\;
	T_{ik} =  \Phi^{-1}_i(\Lambda^{-1}(\btau_k))	
\label{eqn:pqTasFuncsOfutau}
\end{align}
yields the formula in Proposition~\ref{prop:transformedDensity} and concludes the proof.
\end{proof}

\subsection{Lipschitz gradients}
\label{sec:LipschitzGradients}
Before turning to the log-posterior used in this paper, it's helpful to assemble a 
few general facts about Lipschitz functions.

\begin{lemma}[Linear combinations of functions with Lipschitz gradient]  \ \newline
\label{lemma:sumOfLipschitz}
	If all members of set of functions $f_j:\R^n \rightarrow \R$ with 
	$j \in \{1,\dots, N\}$ have Lipschitz gradients, then so does any 
	linear combination of the form
	\begin{equation}
		f(x) = \sum_{j=1}^N \beta_j f_j(x),
	\label{eqn:LipschitzSum}
	\end{equation}
	where the $\beta_j \in \R$ are arbitrary real coefficients.
\end{lemma}

\begin{lemma}[Bounded Hessians and Lipschitz gradients]  \ \newline
\label{lemma:bddHessian}
	If a continuous function $f:\R^n \rightarrow \R$ is twice continuously differentiable
	and its Hessian has bounded norm, then $f$ has a Lipschitz gradient.
\end{lemma}
\noindent
Both these results follow readily from the definition of a Lipschitz function and 
from standard results found in, for example, Section~1.2.2 of \cite{Nesterov:2018wj}.

\subsubsection{Proof of Theorem~\ref{thm:A1} via two Lemmas}
\label{sec:LipschitzGradientProof}
\begin{proof}[Proof of Theorem~\ref{thm:A1}]
The negative log-posterior $U(\bu,\btau)$ in Eqn.~\eqref{eqn:iplaUlogTermsUnconstrained} has three kinds of terms:
\begin{enumerate}
	\item those proportional to $\log\left(\Phi^{-1}_i(\Lambda^{-1}(\bu_j))\right)$,
	\item those proportional to $\log\left(\Phi^{-1}_i(\Lambda^{-1}(\btau_r))\right)$ and
	\item those proportional to 
		$\log\left(\sum_{r=1}^d \Phi^{-1}_i(\Lambda^{-1}(\btau_r))
		 	\Phi^{-1}_r(\Lambda^{-1}(\bu_j))\right)$.
\end{enumerate}
In light of Lemma~\ref{lemma:sumOfLipschitz}, all we need prove is
that terms of these three types have Lipschitz gradients. Those that
involve either $\bu_j$ or  $\btau_r$ alone turn out to be
straightforward as they are covered by
Lemma~\ref{lemma:dirichletDensitiesAreLipschitz} below. The terms
involving both $\bu_j$ and $\btau_r$ together require somewhat more
involved calculations, but can be dispatched as is shown in the proof
Lemma~\ref{lemma:qTermsLipschitzGradient} below. Thus we are able to
establish the first of the three properties required by Akyildiz et al.: our
negative log-posterior has a Lipschitz gradient with respect to the
unconstrained coordinates.
\end{proof}

\begin{lemma}[Lipschitz gradient for sums of log-probabilities] \ \newline
\label{lemma:dirichletDensitiesAreLipschitz}
If a function $f:\Delta_d \rightarrow \R$ has the form
\begin{displaymath}
	f(p_1, \dots, p_d) = \sum_{i=1}^d \beta_i \log(p_i)
\end{displaymath}
then the map $\hat{f}:\R^{d-1} \rightarrow \R$ defined by 
$\hat{f} =  f \circ \Phi^{-1} \circ \Lambda^{-1}$
has a Lipschitz gradient
\begin{displaymath}
	\nabla \hat{f} =  \left( \frac{\partial \hat{f}}{\partial u_1}, \dots,  \frac{\partial \hat{f}}{\partial u_{d-1}}\right ).
\end{displaymath}
\end{lemma}

\begin{lemma}[Lipschitz gradients for terms involving both $\btau_r$ and $\bu_j$] \ \newline
\label{lemma:qTermsLipschitzGradient}
Consider the set of $d \times d$ matrices whose columns lie in
$\Delta_d$ and regard its elements as points in $\Delta_d^d$, the
Cartesian product of $d$ copies of the $d$-simplex. 
Define a function $\hat{T}: \R^{(d-1) \times d} \rightarrow \Delta_d^d$ that
maps $d$-tuples of logit-transformed stick-breaking coordinates to
$\Delta_d^d$. If $\btau = (\btau_1, \dots, \btau_d) \in \R^{(d-1) \times d}$ 
is $d$-tuple of unconstrained coordinates and $T = \hat{T}(\btau)$, 
then $\bT_{\star j}$, the $j$-th column of $T$, is given by 
\begin{displaymath}
	\bT_{\star j} = \Phi^{-1} \left( \Lambda^{-1} \left( \btau_j \right) \right).
\end{displaymath}

Now suppose that $f:\Delta_d \rightarrow \R$ is a function of the form
\begin{displaymath}
	f(\bq) = f(q_1, \dots, q_d) = \sum_{i=1}^d \beta_i \log(q_i).
\end{displaymath}
Then the function $\tilde{f}:\R^{d \times(d-1)} \times \R^{d-1} \rightarrow \R$
defined by 
\begin{displaymath}
	\tilde{f}(\btau, \bu) = 
	f\left( \hat{T}(\btau) \Phi^{-1} \left( \Lambda^{-1}(\bu) \right) \right)
	= f(T \bp),
\end{displaymath}
where $T = \hat{T}(\btau)$ and $\bp = \Phi^{-1} \left( \Lambda^{-1}(\bu) \right)$, has Lipschitz gradients with respect to $\bu$ and the $\btau_j$.
\end{lemma}

\subsubsection{A useful observation about $T_{ij} p_j / q_i$}
As an ingredient in the proof of Lemma~\ref{lemma:qTermsLipschitzGradient}  we will
need the following simple result:
\begin{lemma}[Boundedness of ratios of the form $T_{ij}p_j/q_i$] \ \newline
If $\bp$, $T$ and $\bq$ are as in Lemma~\ref{lemma:qTermsLipschitzGradient}, then
\begin{equation}
	0 \leq \frac{T_{ij} p_j}{q_i} \leq 1.
\end{equation}
whenever $q_i > 0$.
\label{lemma:qRatioBdd}
\end{lemma}

\begin{proof}[Proof of Lemma~\ref{lemma:qRatioBdd}]
	As $\bT_{\star j} \in \Delta_d$, we know that $T_{ij} \geq 0$.
	Similarly, $p_j \geq 0$ and so
	\begin{displaymath}
		0 \leq T_{ij} p_j \leq \sum_{k = 1}^d T_{ik} p_k = q_i.
	\end{displaymath}
	If $q_i > 0$, we can divide these inequalities by it to obtain
	\begin{displaymath}
		0 \leq \frac{T_{ij} p_j}{q_i} \leq 1,
	\end{displaymath}
	which is the result we sought.
\end{proof}

\subsubsection{Proofs of Lemmas~\ref{lemma:dirichletDensitiesAreLipschitz} and
\ref{lemma:qTermsLipschitzGradient}}
\begin{proof}[Proof of Lemma~\ref{lemma:dirichletDensitiesAreLipschitz}]
	We can compute the gradient $\nabla \hat{f}$ using the chain rule and the Jacobians from
	Eqns.~\eqref{eqn:sbPhiToPJacobian} and \eqref{eqn:uJacobianDiag}.
	\begin{equation}
		\nabla \hat{f} = 
		\left[  \frac{\partial \hat{f}}{\partial u_1}, \dots,  \frac{\partial \hat{f}}{\partial u_{d-1}} \right] =
		\nabla_\bp f J_\bphi J_\bu = 
		\left[  \frac{\partial f}{\partial p_1}, \dots,  \frac{\partial f}{\partial p_{d}} \right] J_\bphi J_\bu.
	\end{equation}
	Eqn.~\eqref{eqn:JphiJu} gives the product $J_\bphi J_\bu$ , which we can then act with
	from the right on $\nabla_\bp f$ to obtain
	\begin{align}
		\nabla_\bp f J_\bphi J_\bu  
				& = \left[  \frac{\beta_1}{p_1}, \dots,  \frac{\beta_d}{p_d} \right]
				\left[ \begin{array}{cccc}
					p_1 (1 - \phi_1) & 0 & \cdots & 0 \\ 
					-p_2 \phi_1 & p_2 (1 - \phi_2) &  &  \vdots \\ 
					\vdots & & \ddots & \vdots  \\
					-p_{d-1} \phi_1 & \cdots & & p_{d-1}( 1 - \phi_{d-1}) \\
					- p_d \phi_1 & \cdots & & - p_d \phi_{d-1}
				\end{array} \right], \nonumber \\
				& = \left[ \beta_1 - \phi_1 \left(\sum_{i=1}^d \beta_i \right), 
						\dots,  \beta_{d-1} - \phi_{d-1} (\beta_{d-1} + \beta_d)
					\right].
	\label{eqn:logDirichletDensityGradientJphiJu}
	\end{align}
	In general, the $j$-th component of $\nabla \hat{f}$ is given by
	\begin{align}
		\frac{\partial \hat{f}}{\partial u_j} 
			& = \beta_j - \phi_j \left( \sum_{i=j}^d \beta_i \right) \nonumber \\
			& = \beta_j - \frac{e^{u_j}}{1 + e^{u_j}} \left( \sum_{i=j}^d \beta_i \right).
	\label{eqn:logDirichletDensityGradientWrtU}
	\end{align}
	
	We will now establish that the Hessian of $\hat{f}$ with respect to the $u_j$ has
	bounded norm.  Begin by noting that $\partial \hat{f}/\partial u_j$ depends 
	only on $u_j$, so the Hessian of $\hat{f}$, whose entries are 
	$H_{jk} = \frac{\partial^2 f}{\partial u_j \partial u_k}$, is diagonal. 
	Calculations similar to those leading up to Eqn.~\eqref{eqn:uJacobianDiag} 
	then establish that 
	\begin{align*}
		\frac{\partial^2 \hat{f}}{\partial u_j^2} & = 
			\frac{d}{d u_j} \left( \beta_j - \frac{e^{u_j}}{1 + e^{u_j}} \left( \sum_{i=j}^d \beta_i \right) \right) \\
			& = -\frac{e^{u_j}}{(1 + e^{u_j})^2} \sum_{i=j}^d \beta_i   \\
			& = -\phi_j (1 -\phi_j) \sum_{i=j}^d \beta_i.
	\end{align*}
	Now, $\phi_j (1 -\phi_j)$ is a quadratic with a maximum at $\phi_j = 1/2$ and so the 
	diagonal entries of the Hessian are bounded by 
	\begin{displaymath}
		\left|\frac{\partial^2 f}{\partial u_j^2} \right| 
		\; \leq \; \frac{1}{2} \left(1 - \frac{1}{2} \right) \left| \sum_{i=j}^d \beta_i \right|   
		\; \leq \; \frac{1}{4} \sum_{i=j}^d |\beta_i |	
		\; \leq \;  \frac{1}{4} \sum_{i=1}^d |\beta_i |,
	\end{displaymath}
	which establishes that, for an arbitrary vector $\bx \in \R^{d-1}$,  
	\begin{equation}
		\| H \bx \| \leq \frac{1}{4} \left( \sum_{i=1}^d |\beta_i | \right) \|\bx \|.
		\label{eqn:HessianNormBound}
	\end{equation}
	As $\hat{f}$ is continuous and twice continuously differentiable, 
	Eqn.~\eqref{eqn:HessianNormBound}, along with Lemma~\ref{lemma:bddHessian},
	establishes that $\nabla \hat{f}$ is Lipschitz with constant
	\begin{displaymath}
		K = \frac{1}{4} \sum_{i=1}^d |\beta_i |.
	\end{displaymath}
\end{proof}

\begin{proof}[Proof of Lemma~\ref{lemma:qTermsLipschitzGradient}]
We will need to refer to variously transformed versions of the unconstrained 
coordinates $\bu \in \R^{d-1}$ and $\btau_k \in \R^{d-1}$, so we introduce the following notation:
\begin{displaymath}
	\bp = \Phi^{-1}\left(\bphi \right), \; \; \bphi = \Lambda^{-1}\left(\bu\right)
\end{displaymath}
where $\bp \in \Delta_d$ is a point in the $d$-simplex and $\bphi \in [0,1]^{d-1}$ is the associated set of stick-breaking coordinates, while
\begin{displaymath}
	\bT_{\star k} = \Phi^{-1}\left(\bpsi_k\right), \; \; 
	\bpsi_k = \Lambda^{-1}\left(\btau_k\right)
\end{displaymath}
where $\bT_{\star k} \in \Delta_d$ is both the $k$-th column of $T$ and a 
point in the $d$-simplex while  $\bpsi_k \in [0,1]^{d-1}$ is the associated a set of stick-breaking coordinates. As in the note transmission model in Eqn.~\eqref{eqn:qTp}, 
we will write
\begin{displaymath}
	\bq = T \bp
\end{displaymath}
and note that as we are interested here in values of $\bp$, $T$ and $\bq$ associated
with finite unconstrained coordinates, we can be sure that, for all $j, k \in \{1,\dots,d\}$
\begin{equation}
	0 < p_j < 1, \; 0 < T_{jk} < 1 \mbox{ and } 0 < q_j < 1.
\label{eqn:pqTInInterior}
\end{equation}

With these definitions,  calculations similar to those in 
Eqns.~\eqref{eqn:JphiJu}--\eqref{eqn:logDirichletDensityGradientWrtU}  lead to
\begin{equation}
	\frac{\partial q_i}{\partial u_j} = 
	\frac{\partial}{\partial u_j} \left( \sum_{r = 1}^d T_{ir} p_r \right)  = 
	\sum_{r = 1}^d T_{ir} \, \frac{\partial p_r}{\partial u_j} = 
	T_{ij}p_j - \phi_j \sum_{r = j}^d T_{ir}p_r 
\label{eqn:partialQiWrtUj}
\end{equation}
and
\begin{align}	
	\frac{\partial q_i}{\partial \tau_{jk}} 
		= \frac{\partial}{\partial \tau_{jk}} \left( \sum_{r = 1}^d T_{ir} p_r \right) 
		& =\sum_{r = 1}^d \frac{\partial T_{ir}}{\partial \tau_{jk}} \, p_r \nonumber \\
		& = \frac{\partial T_{ik}}{\partial \tau_{jk}}  p_k  
		= \left\{ \begin{array}{ccl}
			- \psi_{jk} T_{ik}  p_k & \; & \mbox{if $j < i$} \\
			(1 - \psi_{jk}) T_{jk}  p_k && \mbox{if $j = i$} \\
			0 && \mbox{if $j > i$} \\
		\end{array} \right..
\label{eqn:partialQiWrtTauJk}
\end{align}
The components of the gradient of $\tilde{f}$ with respect to the unconstrained 
variables are thus
\begin{align}
	\frac{\partial \tilde{f}}{\partial u_j}  & = 
	\sum_{i = 1}^d \beta_i \left( \frac{\partial \log(q_i)}{\partial u_j} \right) \nonumber \\
	& = \sum_{i = 1}^d \frac{\beta_i}{q_i} \left( \frac{\partial q_i}{\partial u_j} \right) \nonumber \\
	& = 
	\sum_{i = 1}^d \beta_i \left(\frac{T_{ij}p_j}{q_i} 
			- \phi_j \sum_{r = j}^d \frac{T_{ir}p_r}{q_i} \right)
\label{eqn:partialFtildeWrtUj}
\end{align}
and
\begin{equation}	
	\frac{\partial \tilde{f}}{\partial \tau_{jk}} = 
	\sum_{i = 1}^d \beta_i \left( \frac{\partial \log(q_i)}{\partial \tau_{jk}} \right) = 
	\beta_j\frac{T_{jk} p_k}{q_j} - 
		\psi_{jk}\sum_{i =1}^j \beta_i \left(\frac{T_{ik} p_k}{q_i} \right)
\label{eqn:partialFtildeWrtTauIj}
\end{equation}
Our goal now is to show that this gradient is Lipschitz
using Lemma~\ref{lemma:bddHessian}. We'll do this by computing the second
partial derivatives of $\tilde{f}$ and demonstrating that they, and thus all entries
in $\tilde{f}$'s Hessian, are bounded. This, in turn, will imply that the Hessian
has bounded norm, completing the proof.

Consider first the terms $\phi_j T_{ir} p_r / q_i$ that appear in 
Eqn.~\eqref{eqn:partialFtildeWrtUj}. We will prove that they have bounded partial derivatives with respect to the unconstrained coordinates. 
\begin{equation}
	\frac{\partial}{\partial u_k}\left( \frac{\phi_j T_{ir}p_r}{q_i} \right)  = 
		\frac{\partial \phi_j}{\partial{u_k}} \left( \frac{T_{ir} p_r}{q_i} \right) +
		\left( \frac{\phi_j T_{ir}}{q_i} \right) \frac{\partial p_r}{\partial u_k} - 
		\left( \frac{ \phi_j T_{ir}p_r}{q_i^2} \right) \frac{\partial q_i}{\partial u_k} 
\label{eqn:partialQRatioWrtUk}
\end{equation}
The first of the terms at right in Eqn.~\eqref{eqn:partialQRatioWrtUk} is 
\begin{equation}
	\frac{\partial \phi_j}{\partial{u_k}} \left( \frac{T_{ir} p_r}{q_i} \right) = 
	\left\{  \begin{array}{ccl}
		0 & \; & \mbox{if $j \neq k$} \\
		\phi_j(1 - \phi_j) \left(  T_{ir} p_r/ q_i  \right) &&  \mbox{if $j = k$} \\
	\end{array} \right..
\label{eqn:partialQRatioWrtUkFirstTerm}
\end{equation}
We argued in the proof of Lemma~\ref{lemma:dirichletDensitiesAreLipschitz} that  
$0 \leq \phi_j(1 - \phi_j) \leq 1/4$ and we know from Eqn.~\eqref{eqn:pqTInInterior} and Lemma~\ref{lemma:qRatioBdd} that $T_{ir p_r} p_r / q_i \leq 1$ and so the whole term is bounded by 1.

The next term in Eqn.~\eqref{eqn:partialQRatioWrtUk} is 
\begin{equation}
	\left( \frac{\phi_j T_{ir}}{q_i} \right) \frac{\partial p_r}{\partial u_k} = 
	\left\{  \begin{array}{ccl}
		- \phi_k \left(\phi_j T_{ir} p_r/q_i \right) & \; &  \mbox{if $k < r$} \\
		(1 - \phi_k) \left(\phi_j T_{ik} p_k/q_i \right) &&  \mbox{if $k = r$} \\
		0 && \mbox{if $k > r$} 
	\end{array} \right..
\label{eqn:partialQRatioWrtUkSecondTerm}
\end{equation}
where we have used Eqn.~\eqref{eqn:JphiJu}.
Here too arguments depending on the boundedness of $\phi_k$, and Lemma \ref{lemma:qRatioBdd} establish that this
term too is bounded by 1. Finally we turn to the third term in Eqn.~\eqref{eqn:partialQRatioWrtUk}
\begin{align}	
	\left( \frac{ \phi_j T_{ir}p_r}{q_i^2} \right) \frac{\partial q_i}{\partial u_k} 
	& = \left( \frac{ \phi_j T_{ir}p_r}{q_i^2} \right) 
	\left( T_{ik}p_k - \phi_k \sum_{r = k}^d T_{ir}p_r \right) \nonumber \\ 
	& = \left( \frac{ \phi_j T_{ir}p_r}{q_i} \right) 
	\left( \frac{T_{ik} p_k}{q_i} - \phi_k \sum_{r = k}^d \frac{T_{ir}p_r}{q_i} \right) 
\label{eqn:partialQRatioWrtUkThirdTerm}
\end{align}
where I have used Eqn.~\eqref{eqn:partialQiWrtUj}. Arguments similar to those used for the
previous two terms establish that this term is bounded by $(d+1)$, and thus so is the second partial derivative
in Eqn.~\eqref{eqn:partialQRatioWrtUk}. Essentially the same reasoning also proves that
\begin{displaymath}
	\frac{\partial}{\partial u_k}\left( \frac{T_{ir}p_r}{q_i} \right) 
	\mbox{ is bounded $\forall k$}
\end{displaymath}
and this observation, along with Eqn.~\eqref{eqn:partialFtildeWrtUj},
establishes that 
\begin{equation}
	\frac{\partial^2 \tilde{f}}{\partial u_j \partial u_k} \mbox{ and }
	\frac{\partial^2 \tilde{f}}{\partial u_j \partial \tau_{kr}} \mbox{ and }
	\mbox{ are bounded $\forall j, k, r$.}
\label{eqn:fTildeUuDerivs}
\end{equation}

To complete our proof of the boundedness of the entries in $\tilde{f}$'s Hessain, we
need to consider
\begin{equation}
	\frac{\partial}{\partial \tau_{rs}}\left( \frac{\psi_{jk} T_{ik}p_k}{q_i} \right) = 
	\frac{\partial \psi_{jk}}{\partial \tau_{rs}}\left( \frac{T_{ik}p_k}{q_i} \right) +
	\frac{\partial T_{ik}}{\partial \tau_{rs}}\left( \frac{\psi_{jk}p_k}{q_i} \right) -
	\left( \frac{\psi_{jk} T_{ik}p_k}{q_i^2} \right) 
		\frac{\partial q_i}{\partial \tau_{rs}} 
\label{eqn:partialQRatioWrtTauRs}
\end{equation}
The analogues of Eqns.~\eqref{eqn:partialQRatioWrtUkFirstTerm}--\eqref{eqn:partialQRatioWrtUkThirdTerm}
are
\begin{equation}
	\frac{\partial \psi_{jk}}{\partial \tau_{rs}}\left( \frac{T_{ik}p_k}{q_i} \right)
	\left\{  \begin{array}{ccl}
		0 & \; & \mbox{if $s \neq k$ or $r \neq j$} \\
		\psi_{jk}(1 - \psi_{jk}) \left( T_{ik}p_k/ q_i  \right) &&  
			\mbox{if $s = k$ and $r = j$} \\
	\end{array} \right.,
\label{eqn:partialQRatioWrtTauRsFirstTerm}
\end{equation}
\begin{equation}
	\frac{\partial T_{ik}}{\partial \tau_{rs}}\left( \frac{\psi_{jk}p_k}{q_i} \right) =
	\left\{  \begin{array}{ccl}
		- \psi_{rk} \left( \psi_{jk} T_{ik} p_k / q_i\right) & \; &  
				\mbox{if $s = k$ and $r < i$} \\
		(1 - \psi_{ik} )\left( \psi_{jk} T_{ik} p_k / q_i \right) &&  
				\mbox{if $s = k$ and $r = i$} \\
		0 && \mbox{otherwise} 
	\end{array} \right.
\label{eqn:partialQRatioWrtTauRsSecondTerm}
\end{equation}
and
\begin{equation}
	\left( \frac{\psi_{jk} T_{ik}p_k}{q_i^2} \right) 
		\frac{\partial q_i}{\partial \tau_{rs}}
	= \left\{ \begin{array}{ccl}
	- \psi_{rs} (T_{is}  p_s/q_i)(\psi_{jk} T_{ik}p_k / q_i) & \; & \mbox{if $r < i$} \\
	(1 - \psi_{is}) (T_{is}  p_s/q_i)(\psi_{jk} T_{ik}p_k / q_i) && \mbox{if $r = i$} \\
		0 && \mbox{if $r > i$} \\
	\end{array} \right.
\label{eqn:partialQRatioWrtTauRsThirdTerm}
\end{equation}
and, as above, consideration of the boundedness of stick-breaking coordinates,
Eqn.~\eqref{eqn:pqTInInterior} and Lemma~\ref{lemma:qRatioBdd} establish the boundedness of the derivative in Eqn.~\eqref{eqn:partialQRatioWrtTauRs}. 
Thus all entries in the Hessian of $\tilde{f}$ are bounded and so, using 
Lemma~\ref{lemma:bddHessian}, we can conclude that the terms in $U(\bu,\btau)$ 
that involve both $\btau_r$ and $\bu_j$ have Lipschitz gradients.
\end{proof}

\subsection{Strong convexity}
\label{sec:noStrongConvexity}
\begin{proof}[Proof of Theorem~\ref{thm:A2}]
	Lemmas~\ref{lemma:dirichletDensitiesNotStronglyConvex} and
	\ref{lemma:qTermsNotStronglyConvex} below establish that none of the
	three kinds of terms (see the proof of Theorem~\ref{thm:A1})
	that contribute to $U(\bu,\btau)$  are strongly convex and so neither
	is their sum.
\end{proof}

\begin{lemma}[Linear combinations of log-probabilities are not strongly convex] \ \newline
\label{lemma:dirichletDensitiesNotStronglyConvex}
If a function $f:\Delta_d \rightarrow \R$ has the form
\begin{displaymath}
	f(p_1, \dots, p_d) = \sum_{i=1}^d \beta_i \log(p_i)
\end{displaymath}
then the map $\hat{f}:\R^{d-1} \rightarrow \R$ defined by 
$\hat{f} =  f \circ \Phi^{-1} \circ \Lambda^{-1}$ is such that
for every real number $\alpha > 0$ and displacement $\bdelta \in \R^{d-1}$, 
there exists $\bu \in \R^{d-1}$ such that
\begin{displaymath}
	\left\langle 
		\bdelta, 
		\nabla \hat{f}(\bu) - \nabla \hat{f}(\bu - \bdelta) 
	\right\rangle 
	< \alpha \|\bdelta\|^2.
\end{displaymath} 
\end{lemma}

\begin{lemma}[Terms involving both $\btau_k$ and $\bu_j$ are not strongly convex] \ \newline
\label{lemma:qTermsNotStronglyConvex}
If $\bp$, $\bu$, $\btau$, $\bT$ and $\hat{T}$ are as in Lemma~\ref{lemma:qTermsLipschitzGradient}
and $f:\Delta_d \rightarrow \R$ is a function of the form
\begin{displaymath}
	f(\bq) = f(q_1, \dots, q_d) = \sum_{i=1}^d \beta_i \log(q_i),
\end{displaymath}
then the function $\tilde{f}:\R^{(d-1) \times d} \times \R^{d-1} \rightarrow \R$
defined by 
\begin{displaymath}
	\tilde{f}(\btau, \bu) = 
	f\left( \hat{T}(\btau) \Phi^{-1} \left( \Lambda^{-1}(\bu) \right) \right)
	= f(T \bp),
\end{displaymath}
where $T = \hat{T}(\btau)$ and $\bp = \Phi^{-1} \left( \Lambda^{-1}(\bu) \right)$, is
such that for every real number $\alpha > 0$ and all displacements
$\delta\btau \in \R^{(d-1) \times d}$ and $\delta\bu \in \R^{d-1}$, there exist 
$\btau \in \R^{(d-1) \times d}$ and $\bu \in \R^{d-1}$ such that
\begin{displaymath}
	\left\langle 
		\bdelta, 
		\nabla \tilde{f}(\btau, \bu) - \nabla \tilde{f}(\btau - \delta\btau, \bu - \delta\bu) 
	\right\rangle 
	< \alpha \left( \|\delta\btau\|^2 + \|\delta\bu\|^2\right).
\end{displaymath}
\end{lemma}

\begin{proof}[Proof of Lemma~\ref{lemma:dirichletDensitiesNotStronglyConvex}] \ \newline
Assume for contradiction that there is some $\alpha > 0$ such that, for all 
$\bu, \bdelta \in \R^{d-1}$, 
\begin{equation}
	\left\langle 
		\bdelta, 
		\nabla \hat{f}(\bu) - \nabla \hat{f}(\bu - \bdelta) 
	\right\rangle 
	\geq \alpha \|\bdelta\|^2.
\label{eqn:StrongConvexityForDirichletDensity}
\end{equation} 
Eqn.~\eqref{eqn:logDirichletDensityGradientWrtU} tells us that 
components of $\nabla \hat{f}$ are
\begin{displaymath}
	\frac{\partial \hat{f}}{\partial u_j} 
	 = \beta_j - \frac{e^{u_j}}{1 + e^{u_j}} \left( \sum_{i=j}^d \beta_i \right),
\end{displaymath}
so that the components of $\nabla \hat{f}(\bu) - \nabla \hat{f}(\bu - \bdelta)$
are
\begin{displaymath}
	\frac{\partial}{\partial u_j} \left( \hat{f}(\bu) - \hat{f}(\bu - \bdelta) \right)
	= \left( 
		\frac{e^{u_j- \delta_j}}{1 + e^{u_j - \delta_j}} -
		\frac{e^{u_j}}{1 + e^{u_j}} 
	\right)\left( \sum_{i=j}^d \beta_i \right).
\end{displaymath}
Substituting this into the inner product in Eqn.~\eqref{eqn:StrongConvexityForDirichletDensity} yields
\begin{align}
	\left\langle 
		\bdelta, 
		\nabla \hat{f}(\bu) - \nabla \hat{f}(\bu - \bdelta) 
	\right\rangle 
	& = \sum_{j=1}^{d-1} \delta_j 
		\left(
		\frac{e^{u_j- \delta_j}}{1 + e^{u_j - \delta_j}} -
		\frac{e^{u_j}}{1 + e^{u_j}} 
		\right)\left( \sum_{i=j}^d \beta_i \right) \nonumber \\
	& = \sum_{j=1}^{d-1} \delta_j e^{u_j}
		\left(
		\frac{e^{- \delta_j}}{1 + e^{u_j - \delta_j}} -
		\frac{1}{1 + e^{u_j}} 
		\right)\left( \sum_{i=j}^d \beta_i \right) 
\label{eqn:DirichletDensityConvexityContradiction}	
\end{align}
For fixed $\bdelta$, we can make the terms in this sum arbitrarily small by 
sending $u_j$~to~$-\infty$ and so the inner product itself can be made arbitrarily small,
contradicting the inequality in Eqn.~\eqref{eqn:StrongConvexityForDirichletDensity}.
\end{proof}

\begin{proof}[Proof of Lemma~\ref{lemma:qTermsNotStronglyConvex}] \ \newline
Assume for contradiction that there is some $\alpha > 0$ such that, for all 
$\btau, \delta\btau \in \R^{(d-1) \times d}$ and $\bu, \delta\bu \in \R^{d-1}$, 
\begin{equation}
	\left\langle 
		\bdelta, 
		\nabla \tilde{f}(\btau, \bu) - \nabla \tilde{f}(\btau - \delta\btau, \bu - \delta\bu) 
	\right\rangle 
	\geq \alpha \|\bdelta\|^2.
\label{eqn:StrongConvexityForQTerms}
\end{equation} 
Eqns.~\eqref{eqn:partialFtildeWrtUj} and \eqref{eqn:partialFtildeWrtTauIj} tell
us that the components of $\nabla \tilde{f}$ are
\begin{displaymath}
	\frac{\partial \tilde{f}}{\partial u_j} =  
	\sum_{i = 1}^d \beta_i \frac{T_{ij}p_j}{q_i} 
		- \frac{e^{u_j}}{1 + e^{u_j}} 
		\sum_{r = j}^d \beta_i \left(\frac{T_{ir}p_r}{q_i} \right)
\end{displaymath}
and
\begin{displaymath}
	\frac{\partial \tilde{f}}{\partial \tau_{jk}} = 
	\beta_i\frac{T_{ik} p_k}{q_i} - 
		\frac{e^{\tau_{jk}}}{1 + e^{\tau_{jk}}} 
		\sum_{s =1}^j \beta_s \left(\frac{T_{sk} p_k}{q_s} \right).
\end{displaymath}
As in the proof of Lemma~\ref{lemma:dirichletDensitiesNotStronglyConvex}, we can
make the contributions to the inner product from those terms proportional 
to $e^{u_j}/(1 + e^{u_j})$ or 
$e^{\tau_{jk}}/(1 + e^{\tau_{jk}})$ arbitrarily small by sending the $u_j$ or
$\tau_{jk}$ toward~$-\infty$, but the terms in the first sum, which
are of the from $\beta_i T_{ij}p_j/q_i$, 
require more careful treatment.

First, note that Eqns.~\eqref{eqn:stickBreakingPhiToP},  
\eqref{eqn:invLogitTransform} and \eqref{eqn:oneMinusPhi} imply that
\begin{equation}
	p_i = \frac{e^{u_i}}{\prod_{j=1}^i (1 + e^{u_j})} \;\;
	\mbox{ for $1 \leq j < d$ and } \;\;
	p_d = \frac{1}{\prod_{j=1}^d (1 + e^{u_j})}
\label{eqn:PiAsFuncOfUs}
\end{equation}
and 
\begin{equation}
	T_{ik} = \frac{e^{\tau_{ik}}}{\prod_{j=1}^i (1 + e^{\tau_{jk}})} \;\;
	\mbox{ for $1 \leq j < d$ and } \;\;
	T_{dk} = \frac{1}{\prod_{j=1}^d (1 + e^{\tau_{jk}})}.
\label{eqn:TikAsFuncOfTaus}
\end{equation}
As we are interested in gradients with respect to the unconstrained coordinates,
the terms we need to consider involve $T_{ij}p_j/q_i$ with 
$1 \leq i, j < d$ and so we can write them as follows:

\begin{align}
	\frac{T_{ij}p_j}{q_i} & = \frac{T_{ij}p_j}{\sum_{k=1}^{d} T_{ik} p_k} \nonumber \\
	& = \frac{T_{ij}p_j}{\sum_{k=1}^{d-1} T_{ik}p_k + T_{id} p_d} \nonumber \\
	& = \left( 
		\frac{ 
			\frac{e^{\tau_{ij}}}{\prod_{r=1}^i (1 + e^{\tau_{rj}})} 
			\frac{e^{u_j}}{\prod_{r=1}^j (1 + e^{u_r})}
		}{
			\sum_{k=1}^{d-1} 
			\frac{e^{\tau_{ik}}}{\prod_{r=1}^i (1 + e^{\tau_{rk}})}
			\frac{e^{u_k}}{\prod_{r=1}^k (1 + e^{u_r})}
			+ 
			\frac{1}{\prod_{r=1}^{d-1} (1 + e^{\tau_{rk}}) (1 + e^{u_r})}
		}
	\right) \nonumber \\ 
	& = e^{\tau_{ij} + u_j}
	\left( 
		\frac{ 
			\frac{1}{\prod_{r=1}^i (1 + e^{\tau_{rj}})} 
			\frac{1}{\prod_{r=1}^j (1 + e^{u_r})}
		}{
			\sum_{k=1}^{d-1} 
			\frac{e^{\tau_{ik}}}{\prod_{r=1}^i (1 + e^{\tau_{rk}})}
			\frac{e^{u_k}}{\prod_{r=1}^k (1 + e^{u_r})}
			+ 
			\frac{1}{\prod_{r=1}^{d-1} (1 + e^{\tau_{rk}}) (1 + e^{u_r})}
		}
	\right). 
\end{align}
The exponential factor
$e^{\tau_{ij} + u_j}$ can be made arbitrarily small by sending an appropriate
unconstrained coordinate---$\tau_{jk}$ or $u_k$---to $-\infty$. Meanwhile,
the quantity in parentheses is bounded above by 1, so an argument similar 
to that in the proof of  
Lemma~\ref{lemma:dirichletDensitiesNotStronglyConvex} establishes
that the contributions to the inner product in Eqn.~\eqref{eqn:StrongConvexityForQTerms}
from terms of the form $\beta_i T_{ij} p_j/q_i$ can be made arbitrarily small, 
no matter the values of $\delta\bu$ and $\delta\btau$. This provides the contradiction
we sought---the inner product cannot be both arbitrarily small and bounded
from below by a quadratic---and so completes the proof of
Lemma~\ref{lemma:qTermsNotStronglyConvex}.
\end{proof}

\subsection{Bounded second moments}
\label{sec:boundedSecondMoments}

\begin{proof}[Proof of Theorem~\ref{thm:A3}]
	Previous work suggests that pupils learn their tutors' songs faithfully,
	so we initialise the $\btau_j = (\tau_{1j}, \dots, \tau_{dj})$ with 
	unconstrained coordinates that correspond
	to a constant matrix close to the the $d \times d$ identity,
	\begin{displaymath}
		T = \left[ \begin{array}{cccc}
			1 - \epsilon & \frac{\epsilon}{d-1} & \cdots & \frac{\epsilon}{d-1} \\
			\frac{\epsilon}{d-1} & 1 - \epsilon & \cdots & \frac{\epsilon}{d-1} \\
			\vdots & \vdots & \ddots & \vdots \\
			\frac{\epsilon}{d-1} & \frac{\epsilon}{d-1} & \cdots & 1 - \epsilon
		\end{array} \right]
		\; \mbox{ so } \;
		T_{ij} = \left\{ 
			\begin{array}{ccl} 
				(1 - \epsilon) & \; & \mbox{if $i=j$} \\ 
				\epsilon/(d-1) && \mbox{otherwise}
			\end{array} \right..
	\end{displaymath}
	Once $\epsilon$ is fixed (we used $\epsilon =0.1$ in the computations that lead 
	to Figure~\ref{fig:birdsongTransMat}) this is a constant matrix and 
	so the second moments are of the $\tau_{ij,0}$ are trivially bounded by the
	squares of their values.
	
	We initialise the $\bu_j$ by computing the unconstrained coordinates
	corresponding to a value $\bp_j$ drawn from a Dirichlet distribution
	with shape parameters
	\begin{displaymath}
		\alpha_{ij} = x_{ij} + \alpha_p
	\end{displaymath}
	where the $x_{ij}$ give the tutor's note-usage counts at aligned position $j$ and 
	$\alpha_p = 1$ is the parameter of the prior on $\bp_j$.
	
	As we draw the $\bu_{j,0} = (u_{1j,0}, \dots, u_{dj,0})$ independently, it 
	is sufficient to establish that the unconstrained coordinates corresponding
	to a $\bp_j$ drawn from a Dirichlet distribution have bounded second moments.
	Lemma~\ref{lemma:unconstrainedDirichletMoments} below does this, establishing that
	$\E[u_{ij,0}^2]$ is bounded by a quantity that depends only on 
	$\alpha_{ij}$, and so completes the proof of Theorem~\ref{thm:A3}.
\end{proof}

\begin{lemma}[Second moments of unconstrained coordinates for Dirichlet $\bp$] \ \newline
\label{lemma:unconstrainedDirichletMoments}
Let $\bu = (u_1, \dots, u_{d-1}) \in \R^{d-1}$ be the unconstrained coordinates 
corresponding to $\bp = (p_1, \dots, p_d) \in \Delta_d$. If $\bp$ is drawn 
from a Dirichlet distribution with shape parameters 
$\balpha = (\alpha_1, \dots, \alpha_d)$, then
\begin{displaymath}
	\E[u_j^2] \leq \frac{\Gamma(\alpha_j+\gamma_j)}{\Gamma(\alpha_j) \Gamma(\gamma_j)}
	\left( \frac{2}{\alpha_j^3} \; + \; \frac{2}{\gamma_j^3} \right),
\end{displaymath}
where
\begin{displaymath}
	\gamma_j = \sum_{i=j+1}^d \alpha_i.
\end{displaymath}
\end{lemma}

\begin{proof}[Proof of Lemma~\ref{lemma:unconstrainedDirichletMoments}] \ \newline
If $\bu$ and $\bp$ are as in Lemma~\ref{lemma:unconstrainedDirichletMoments} 
then the distribution of $\bu$ has density
\begin{equation}
	f(\bu) = \left[
		\frac{\Gamma\left( \alpha \right)}{\prod_{i=1}^d \Gamma(\alpha_i)}
		\prod_{i=1}^{d} p_i^{\alpha_i-1}
	\right] 
	\left[
		\prod_{i=1}^{d-1} p_i (1- \phi_i) 
		\vphantom{\frac{\Gamma\left( \alpha \right)}{\prod_{i=1}^d \Gamma(\alpha_i)}}
	\right] 
\label{eqn:densityForUZero}
\end{equation}
where
\begin{displaymath}
	\alpha = \sum_{i=1}^d \alpha_i, \qquad p_i = \Phi^{-1}\left( \Lambda^{-1}\left( u_i \right)\right) 
	\mbox{ for $1 \leq i \leq d-1$}
\end{displaymath}
and, using Eqn.~\eqref{eqn:stickBreakingPToPhiToo},
\begin{displaymath}
	p_d = \prod_{i=1}^{d-1}\left(1- \phi_i\right)
	= \prod_{i=1}^{d-1}\left(1- \Lambda^{-1}\left( u_i \right)\right).
\end{displaymath}
The leftmost factor in $f(\bu)$ is the density of the Dirichlet distribution, while the remaining factor comes from the determinant of the Jacobians in Eqn.~\eqref{eqn:JphiJu}.

Rearranging the products in Eqn.~\eqref{eqn:densityForUZero} and using
Eqn.~\eqref{eqn:PiAsFuncOfUs} to write the density explicitly in terms of the 
$u_j$, we find
\begin{align}
	f(\bu) &= 
	\frac{\Gamma\left( \alpha \right)}{\prod_{i=1}^d \Gamma(\alpha_i)}
	\left[
		\prod_{i=1}^{d-1} p_i^{\alpha_i}
	\right] p_d^{\alpha_d}
	\nonumber \\
	& = 
	\frac{\Gamma\left( \alpha \right)}{\prod_{i=1}^d \Gamma(\alpha_i)}
	\left[
		\prod_{i=1}^{d-1} 
		\left( 
			\frac{e^{u_i}}{\prod_{j=1}^i \left(1 +e^{u_j}\right)}
		\right)^{\alpha_i}
	\right] 
	\left( \frac{1}{\prod_{j=1}^{d-1} (1 +e^{u_j})}\right)^{\alpha_d}.
\end{align}
Gathering factors that involve $e^{u_j}$ leads to 
\begin{equation}
	f(\bu) = 
	\frac{\Gamma\left( \alpha \right)}{\prod_{i=1}^d \Gamma(\alpha_i)}
		\left[ \prod_{i=1}^{d-1} 
			\frac{e^{\alpha_i u_i}}{\left(1 +e^{u_i}\right)^{\alpha_i + \gamma_i}}
		\right],
\label{eqn:factoredUnconstrainedDensity}
\end{equation}
where the $\gamma_i$ appearing in the exponents of the denominators are 
\begin{displaymath}
	\gamma_i = \sum_{j=i+1}^d \alpha_j.
\end{displaymath}

Each of the factors in the representation of $f(\bu)$ in 
Eqn.~\eqref{eqn:factoredUnconstrainedDensity} depends on only a single one of the
$u_j$ and so if we define
\begin{displaymath}
	I_j = \int_{-\infty}^{\infty} 
		\frac{e^{\alpha_j u_j}}{\left(1 +e^{u_j}\right)^{\alpha_j + \gamma_j}} \, du_j,
\end{displaymath}
the fact that $f(\bu)$ is properly normalised implies 
\begin{displaymath}
	\int_{\R^{d-1}} f(\bu) \, d\bu = 
	\frac{\Gamma\left(\alpha \right)}{\prod_{i=1}^d \Gamma(\alpha_i)} \,
	\prod_{i=1}^{d-1} I_j = 1.
\end{displaymath}
We can also calculate $I_j$ explicitly by introducing $\phi_j = e^{u_j} / (1 + e^{u_j})$:
\begin{align}
	I_j = \int_{-\infty}^{\infty} 
		\frac{e^{\alpha_j u_j}}{\left(1 +e^{u_j}\right)^{\alpha_j + \gamma_j}} \, du_j
	& = \int_0^1 \phi_j^{\alpha_j} (1 - \phi_j)^{\gamma_j} \, 
		\frac{d\phi_j}{\phi_j (1 - \phi_j)} \nonumber \\
	& = \int_0^1 \phi_j^{\alpha_j-1} (1 - \phi_j)^{\gamma_j-1} \, d\phi_j \nonumber \\
	& = B(\alpha_j, \gamma_j)
	= \frac{\Gamma(\alpha_j) \Gamma(\gamma_j)}{\Gamma(\alpha_j+\gamma_j)},
\label{eqn:IjAsBetaFunction}
\end{align}
where $B(p,q)$ is the Beta function.

Now express the second moment of $u_j$ in terms of the $I_j$:
\begin{align}
	\E[u_j^2] = \int_{\R^{d-1}} u_j^2 \,f(\bu) \, d\bu 
	& = \frac{\int_{-\infty}^{\infty} 
		\frac{u_j^2 \, e^{\alpha_j u_j}}{\left(1 +e^{u_j}\right)^{\alpha_j + \gamma_j}} 
	\, du_j}{I_j} \nonumber \\ 
	& = \frac{\Gamma(\alpha_j+\gamma_j)}{\Gamma(\alpha_j) \Gamma(\gamma_j)}
	\int_{-\infty}^{\infty} 
		\frac{u_j^2 \, e^{\alpha_j u_j}}{\left(1 +e^{u_j}\right)^{\alpha_j + \gamma_j}} 
	\, du_j,
\label{eqn:secondMomentAsRatio}
\end{align}
Using this, we can establish the our result by bounding the integral
at right in the last line of Eqn.~\eqref{eqn:secondMomentAsRatio}:  
\begin{align}
	\int_{-\infty}^{\infty} 
		\frac{u_j^2 \,e^{\alpha_j u_j}}{\left(1 +e^{u_j}\right)^{\alpha_j + \gamma_j}} 
		\, du_j 
	& = \int_{-\infty}^{0} 
		\frac{u_j^2 \,e^{\alpha_j u_j}}{\left(1 +e^{u_j}\right)^{\alpha_j + \gamma_j}} 
		\, du_j \; + \; 
		\int_{0}^{\infty}
		\frac{u_j^2 \,e^{\alpha_j u_j}}{\left(1 +e^{u_j}\right)^{\alpha_j + \gamma_j}} 
		\, du_j \nonumber \\
	& \leq \int_{-\infty}^{0} u_j^2 \,e^{\alpha_j u_j} \, du_j \; + \; 
		\int_{0}^{\infty}
			u_j^2
			\left( \frac{e^{u_j}}{1 +e^{u_j}}\right)^{\alpha_j} 
			\left( \frac{1}{1 +e^{u_j}}\right)^{\gamma_j} 
			\, du_j \nonumber \\
	& \leq \frac{2}{\alpha_j^3} \; + \; 
		\int_{0}^{\infty}
			u_j^2 \,e^{-\gamma_j u_j} \, du_j \nonumber \\
	& \leq \frac{2}{\alpha_j^3} \; + \; \frac{2}{\gamma_j^3}.
\label{eqn:IjUpperBound}
\end{align}
Finally, by combining Eqns.~\eqref{eqn:secondMomentAsRatio} and 
\eqref{eqn:IjUpperBound} we obtain
\begin{displaymath}
	\E[u_j^2] \leq \frac{\Gamma(\alpha_j+\gamma_j)}{\Gamma(\alpha_j) \Gamma(\gamma_j)}
	\left( \frac{2}{\alpha_j^3} \; + \; \frac{2}{\gamma_j^3} \right),
\end{displaymath}
completing the proof of Lemma~\ref{lemma:unconstrainedDirichletMoments}.
\end{proof}

\newpage
\printbibliography

\newpage
\section*{Supplementary Material}
\label{sec:supplement}
\renewcommand{\thetable}{{S\arabic{table}}}
\renewcommand{\thefigure}{{S\arabic{figure}}}
\setcounter{table}{0}
\setcounter{figure}{0}

\begin{figure}[hb]
	\includegraphics[width=0.32\linewidth]{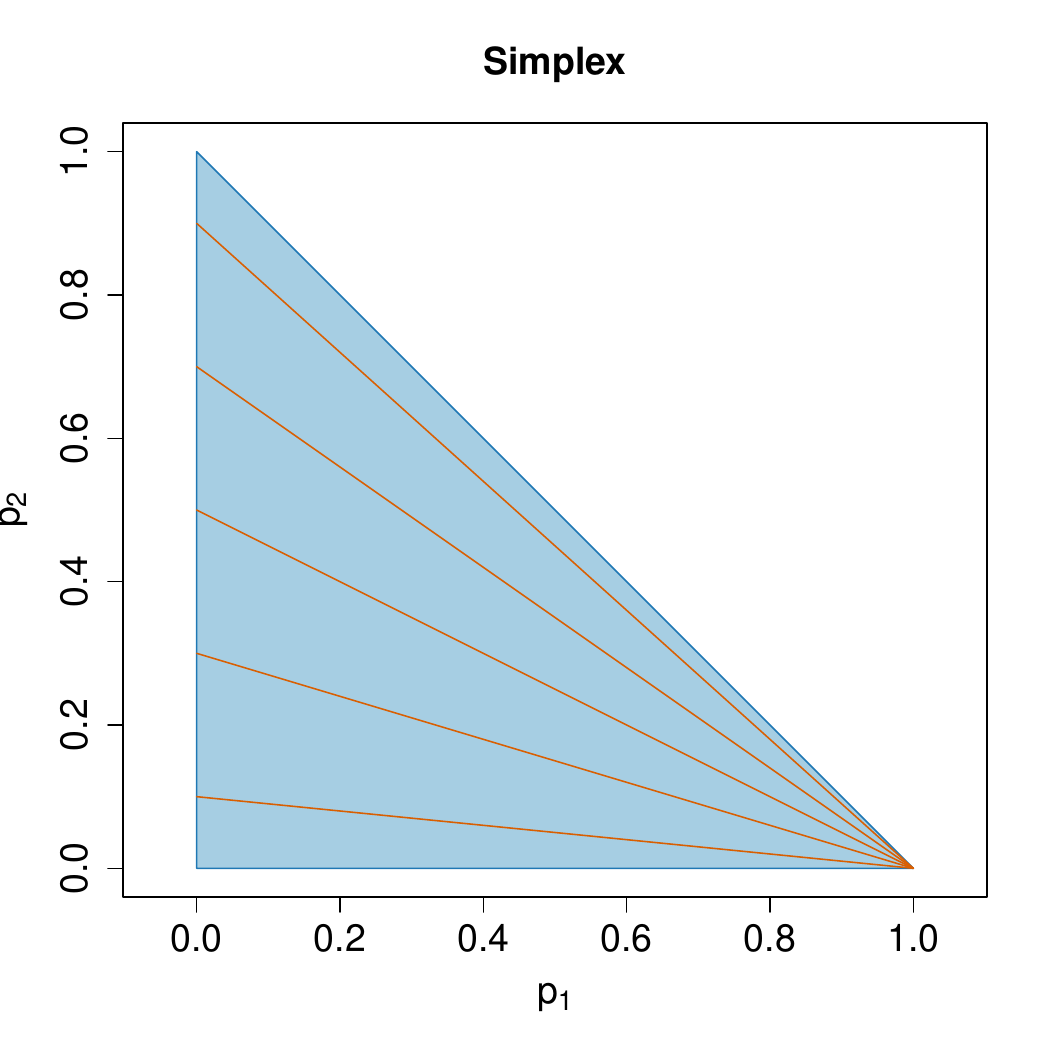}
	\includegraphics[width=0.32\linewidth]{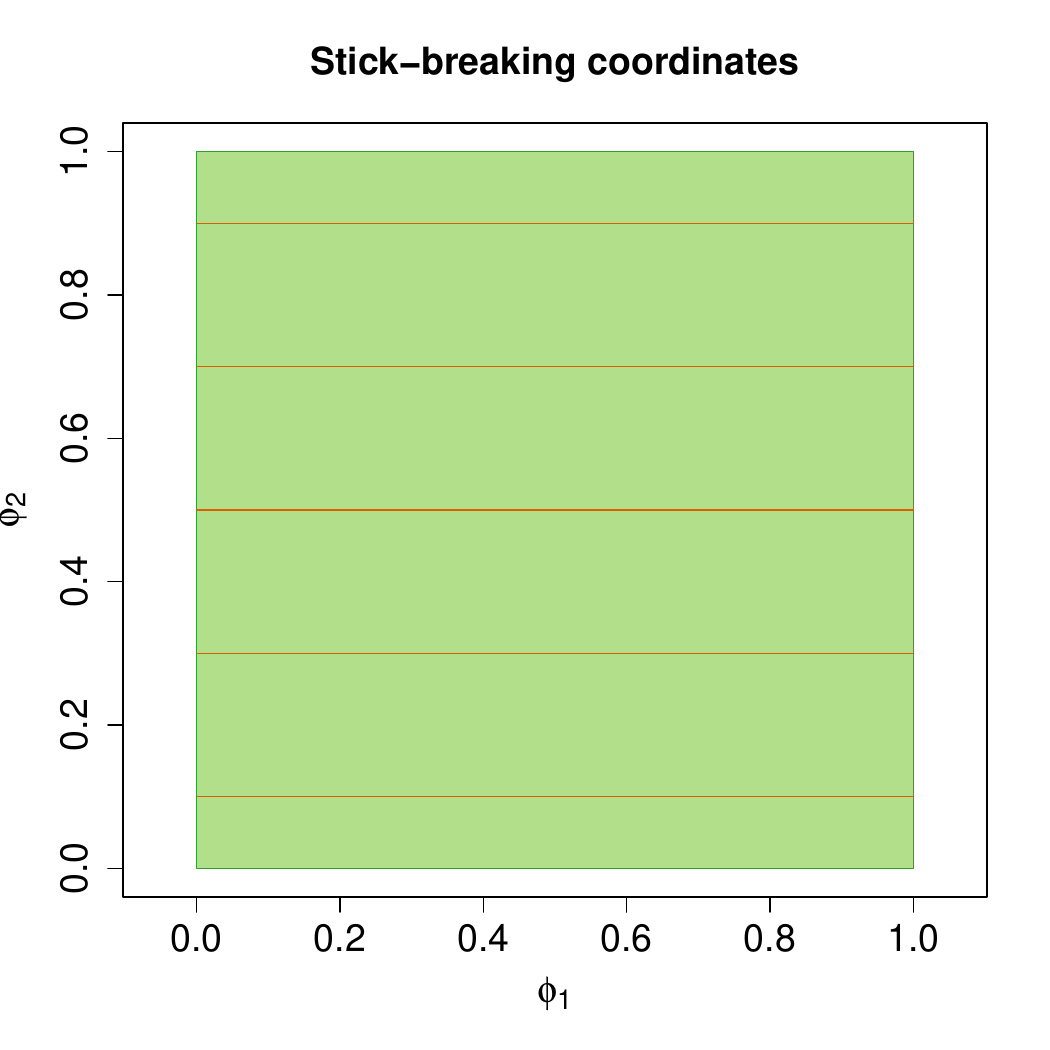}
	\includegraphics[width=0.32\linewidth]{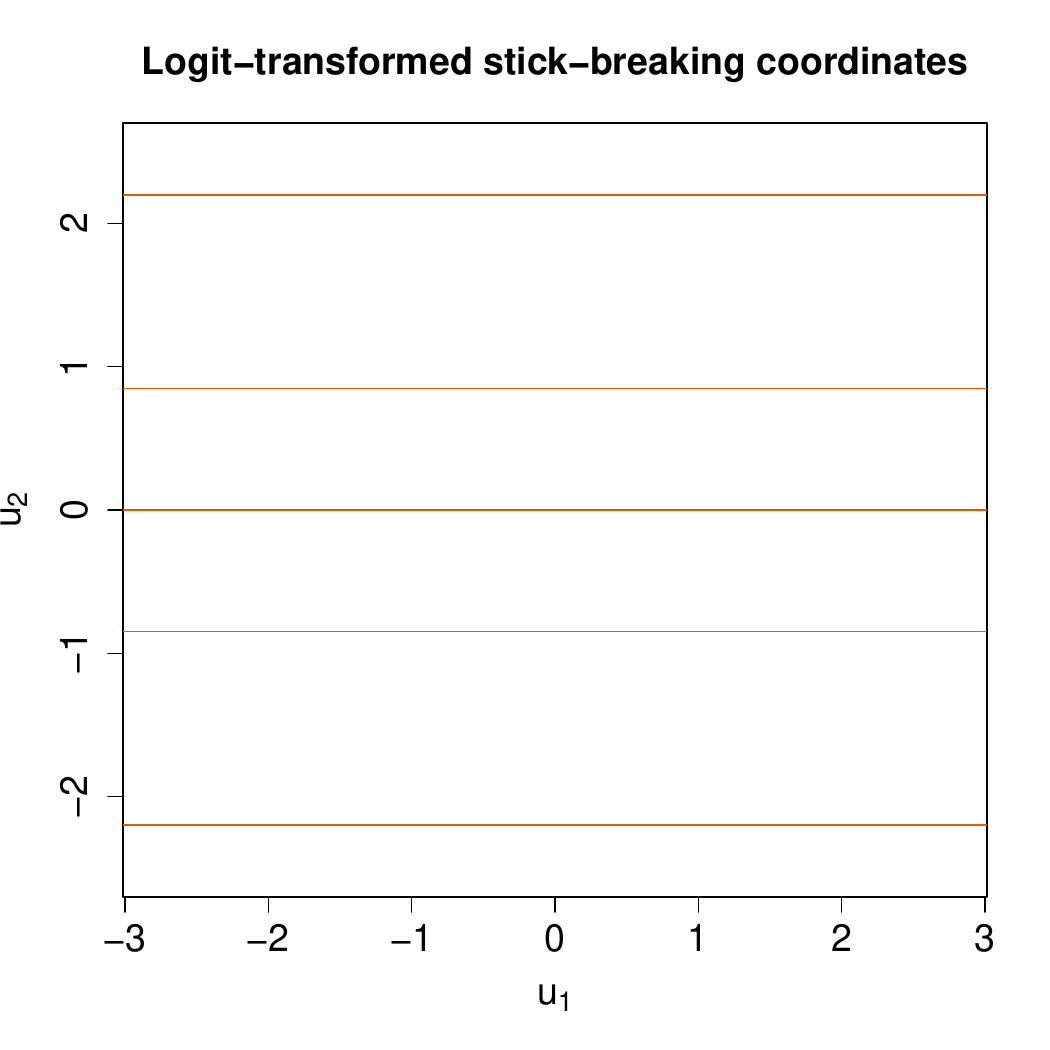}
	\caption{
		A diagram illustrating the coordinates systems used in the paper. The panel 
		at left is the projection of the three-simplex $\Delta_3$ into the $(p_1, p_2)$
		plane, while the panel in the center shows the associated stick-breaking coordinates
		$\phi_1 = p_1$ and $\phi_2 = p_2/(1 - p_1)$. The panel at right shows the 
		logit-transformed coordinates $u_j = \log(\phi_j / (1 - \phi_j))$. All three panels
		have dark orange lines corresponding to the sets $u_2 = c$ for 
		$c \in \{0.1, 0.3, 0.5, 0.7, 0.9\}$.
	}
\label{fig:CoordinateSystemsS1}
\end{figure}

\begin{figure}
	\includegraphics[width=0.95\linewidth]{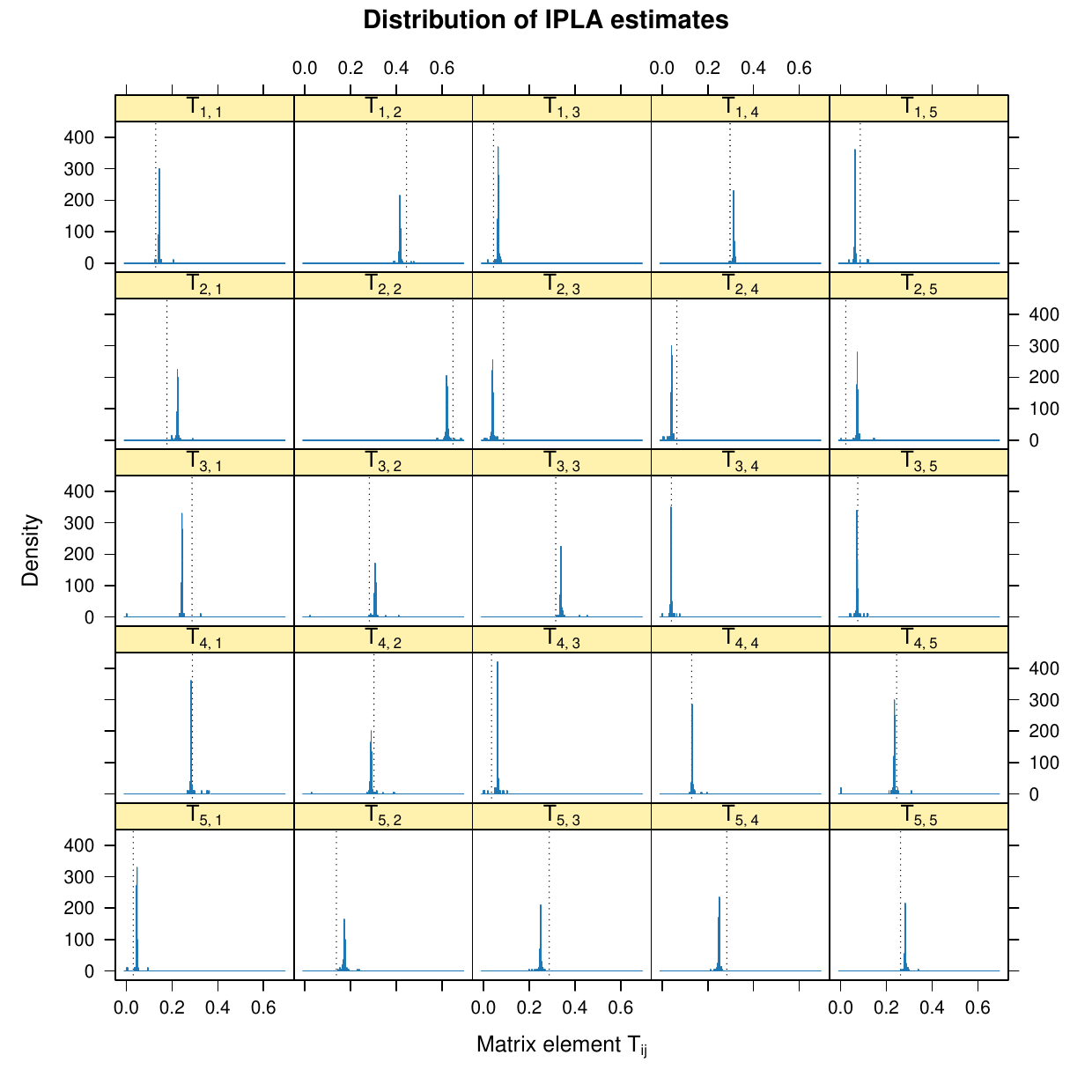}
	\caption{
		Histograms showing the distribution of the estimates of the entries $T^\star_{ij}$ in 
		the MMLE transmission matrix obtained from 200 independently-initialised
		runs of the IPLA. Visual inspection finds little evidence of multimodality
		for any entry, a conclusion supported by the dip test of unimodality developed
		in \cite{Hartigan:1985yn}. This is a hypothesis test whose null hypothesis
		is that the data are drawn from a distribution that has a continuous,
		unimodal density. In all cases, the dip test suggested no significant 
		deviation from unimodality ($p > 0.7$ in all cases).
	}
\label{fig:TijUnimodality}
\end{figure}

\begin{table}[hp]
\centering
\begin{tabular}{lrlrrr}
		 & Number of & Note & Note & Proportion & Cumulative \\ 
	Name & Occurrences & Char. & Class & of all notes & Proportion \\ 
  \hline
  Curve &  8470 & A &     1 & 0.3687 & 0.3687 \\ 
  Slope &  5855 & B &     2 & 0.2549 & 0.6236 \\ 
  Chip &  2946 & C &     3 & 0.1282 & 0.7518 \\ 
  Chirp &  1954 & D &     4 & 0.0851 & 0.8369 \\ 
  Stack &  1741 & E &     5 & 0.0758 & 0.9127 \\ 
  Curve/Stack &   991 & F &     6 & 0.0431 & 0.9558 \\ 
  Whine &   310 & G &     7 & 0.0135 & 0.9693 \\ 
  Curve/Slope &   188 & H &     8 & 0.0082 & 0.9775 \\ 
  Stack/Curve &   143 & I &     9 & 0.0062 & 0.9837 \\ 
  U\_Stack &   132 & J &     9 & 0.0057 & 0.9895 \\ 
  NLP\_Slope &   101 & K &     9 & 0.0044 & 0.9939 \\ 
  Noisy\_Chirp &    97 & L &     9 & 0.0042 & 0.9981 \\ 
  Wavy &    23 & M &     9 & 0.0010 & 0.9991 \\ 
  Wavy\_Short &    10 & N &     9 & 0.0004 & 0.9995 \\ 
  Stack/Slope &     9 & O &     9 & 0.0004 & 0.9999 \\ 
  NLP\_Other &     2 & P &     9 & 0.0001 & 1.0000 \\ 
\end{tabular}
\caption{
	A table summarising the frequency with which the various notes defined
	in \citep{Lewis:2021hj} are sung. The rows of the table correspond to notes
	and are ordered according to decreasing frequency of occurrence in the
	data set. The first column gives the note's name as
	specified by Lewis et al., the second gives the number of occurrences
	of the note across all songs used in our study, the third gives 
	the character used to represent the note in alignments such as that
	illustrated in Figure~4 of the paper and the fourth column gives the note's
	number as used in the transmission matrix $T$: note that the 8 least 
	commonly-sung notes, which account for around 2.25\% of all notes
	recorded, are amalgamated into a single class. The rightmost two columns
	give, respectively, the proportion of all notes that are of the given
	type and the cumulative proportion for the given note and those sung 
	more commonly.
}
\label{tbl:noteUsageS1}
\end{table}
\end{document}